%% file: Paper_RL-v0.tex
\documentclass[final,5p,times,twocolumn,fleqn]{elsarticle}

\journal{International Journal of Electrical Power and Energy Systems
}

\input{preamble_eslvier.tex}

\usepackage{algpseudocode,float}
\usepackage{lipsum}
\usepackage{adjustbox}

\usepackage{graphicx} 
\usepackage{multirow}
\usepackage{tabularx}
\usepackage{booktabs}
\usepackage{rotating, makecell}
\usepackage{lipsum} 
\usepackage{hyperref}

\usepackage{csquotes}

\usepackage{lipsum}
\usepackage{lettrine}

\usepackage{eqnarray,amsmath}

\begin{document}
	\begin{frontmatter}

		\title{{Wide-Area Feedback Control for Renewables-Heavy Power Systems: A Comparative Study of Reinforcement Learning and Lyapunov-Based Design}}
		
		\fntext[label2]{This work is supported by National Science Foundation under Grants 2152450 and 2151571.}
		
		
		\author[inst1]{Muhammad Nadeem}
		
		\affiliation[inst1]{organization={Vanderbilt University},
			addressline={ Civil and Environmental Engineering Department}, 
			city={Nashville},
			postcode={37235}, 
			state={TN},
			country={US}}
		\ead{muhammad.nadeem@vanderbilt.edu}
		\author[inst1]{MirSaleh Bahavarnia}
		\ead{mirsaleh.bahavarnia@vanderbilt.edu}
		\author[inst1]{Ahmad F. Taha}
		\ead{ahmad.taha@vanderbilt.edu}
\vspace{-0.5cm}
\begin{abstract}
As renewable energy sources become more prevalent, accurately modeling power grid dynamics is becoming increasingly more complex.  Concurrently, data acquisition and realtime system state monitoring are becoming more available for control centers. This motivates shifting from \textit{model- and Lyapunov-based} feedback controller designs toward \textit{model-free} ones. Reinforcement learning (RL) has emerged as a key tool for designing model-free controllers. Various studies have been carried out to study voltage/frequency control strategies via RL. However, usually a simplified system model is used neglecting detailed dynamics of solar, wind, and composite loads---and damping system-wide oscillations and modeling power flows are all usually ignored. To that end, we pose an optimal feedback control problem for a detailed renewables-heavy power system, defined by a set of nonlinear differential algebraic equations (NDAE).  The control problem is solved using a completely model-free design via RL as well as using a model-based approach built upon the Lyapunov stability theory with guarantees.  The paper in its essence seeks to explore whether data-driven feedback control should be used in power grids over its model-driven counterpart. Theoretical developments and thorough case studies are presented with an eye on this exploration.  Finally, a detailed analysis is provided to delineate the strengths and weaknesses of both approaches for renewables-heavy grids.
		\end{abstract}
	\vspace{-0.3cm}
		\begin{keyword}
			Reinforcement learning, solar and wind-based power plants, feedback control, Lyapunov stability.
		\end{keyword}
	\end{frontmatter}
\vspace{-0.6cm}
\section{Introduction and motivation}\label{section:intro}
\vspace{-0.1cm}
With the increased penetration of wind and solar-based energy resources, the overall transient stability of power systems is deteriorating. In particular, there is a significant increase in low and ultra-low frequency oscillations (LFOs and ULFOs)  in the future power grids with high penetration of renewables. These LFOs and ULFOs if not properly damped can cause system-wide instabilities and blackouts \cite{gupta2021coordinated}. State/output feedback controllers can play a crucial role in mitigating such oscillations and improving system transient stability after large disturbances. Based on realtime measurements, these feedback controllers can provide additional control signals to the power plants thus improving system robustness toward disturbances \cite{LiuITPWRS2021}.

Traditionally, in the literature, model-based approaches are utilized to design feedback controllers. These models, based on advanced differential-algebraic equations (DAEs)  provide precise, physics-based descriptions of system dynamics, allowing for the design of sophisticated feedback controllers that can effectively regulate frequency, improve LFOs/ULFOs, and other grid parameters \cite{RinaldiICSL2017, SiljakITPWRS2002, ZhangITPWRS2013, HadidiITSG2013}.  However, the performance of model-based control heavily relies on the accuracy of these models, which are inherently limited by assumptions and simplifications. 
Also, as  power networks become more dynamic and uncertain
due to the integration of renewables and other
distributed resources, accurately modeling their dynamics 
becomes an increasingly daunting task. Moreover, even if the accurate model is computed, re-adjusting it in realtime and re-designing the feedback control law (every time the model changes) becomes highly inefficient and impractical.  Thus, there is a growing motivation to transition toward completely model-free feedback control strategies.

Reinforcement learning (RL) is a key tool in the development of model-free feedback controllers due to its self-learning capabilities. RL allows the controller to autonomously learn the optimal control policy through continuous interaction with the environment while achieving predefined objectives. As a result, a variety of RL-based control algorithms have been proposed for power systems in recent years. For example, \cite{cui2022reinforcement, feng2024online} introduces a frequency control algorithm aimed at minimizing both the frequency nadir and the required control inputs. In \cite{WangITPWRS2020}, a voltage control strategy is developed using deep RL, where the RL algorithm minimizes voltage deviations across buses by actively adjusting generator bus voltages. Additionally, various RL-based automatic generation control (AGC) schemes have been proposed to regulate system frequency during transients. These include approaches such as Q-learning, actor-critic-based integral RL, and policy gradient (PG) techniques, as seen in \cite{yan2020multi}, \cite{wang2019multiobjective}, and \cite{singh2017distributed}. Similarly, reactive power control methodologies designed to improve system voltages using different RL approaches are discussed in \cite{yin2021emotional}. 

\textcolor{black}{More recent work has added stability guarantees to RL-based controllers, creating stability-constrained, model-free designs. For example, \cite{shi2022stability} proposes an RL method for real-time voltage control that uses a Lyapunov function to ensure formal voltage stability. In \cite{mukherjee2020distributed}, the authors introduce a distributed RL framework with stability guarantees for interconnected power systems, based on a diagonally dominant eigenvalue condition. Likewise, \cite{yuan2024reinforcement} combines Lyapunov stability theory and neural network controllers in an RL-based transient frequency control scheme. Additionally, \cite{wan2025stability} incorporates Lyapunov stability principles into RL-based power converter control, ensuring converter stability while improving performance. For a broader overview of RL applications in power system control, see the comprehensive surveys in \cite{chen2022reinforcement} and \cite{GLAVIC201922}.}

However, it's important to note that much of the current literature simplifies power system models by (1) neglecting algebraic constraints (power flow and balance equations), (2) using lower-order models for synchronous machine dynamics, and (3) overlooking the complexities of power electronics-based models for solar and wind generation. Moreover, many studies focus exclusively on minimizing frequency and voltage deviations, without addressing the critical issue of damping system-wide low-frequency and ultra-low-frequency oscillations (LFOs/ULFOs). This omission stems from not solving the complete optimal feedback control problem, such as the traditional linear-quadratic regulator (LQR)-type control design. As previously mentioned, addressing LFOs and ULFOs is essential for maintaining grid stability and optimizing power transfer capabilities \cite{gupta2021coordinated}.

Some recent efforts have been made to design various decentralized and centralized damping feedback controllers using RL. For example, in \cite{mukherjee2021scalable} using a reduced power system model a wide-area damping controller (WADC) has been proposed to improve system oscillations after disturbance. Similarly, in \cite{vrabie2009adaptive} authors have developed a WADC using a policy iteration algorithm, while the study \cite{jiang2012robust} has introduced a decentralized feedback controller design via an off-policy iteration-based RL technique. However, these studies also rely on simplified power system models, considering lower-order generator dynamics and neglecting algebraic constraints, dynamics of loads,  wind, and solar power plants. Using a simplified power system model during training can cause serious stability issues as the learned policy might be unstable when applied to the actual power system model with detailed dynamics.

Furthermore, it is clear that designing feedback controllers using a model-free approach via RL has a key advantage as it does not require the knowledge of system parameters, dynamics, or topology. However, some basic questions arise while solving an optimal state feedback control problem for renewables heavy power systems: How computationally efficient is model-free control as compared to model-based design? Does the model-free approach yield a \textit{better} control law as compared to the model-based approach? To that end, this paper addresses the aforementioned literature gaps and research questions by designing model-free and model-based WADCs for  renewables heavy power grids with detailed wind, solar, and composite load dynamics (given in Sec.~\ref{sec: System model}). The problem formulation is given in Sec.~\ref{sec:controller design}. The proposed model-free WADC is based on reinforcement learning (Sec. \ref{sec:RL WAC}) while the model-based controller is designed  via Lyapunov control theory (Sec. \ref{sec:Model-based WAC}), and Sec.~\ref{sec:case studies} presents thorough case studies and some conclusions. 

The technical contributions of this paper are as follows:

\begin{itemize}
	\item We propose WADC for renewables heavy power systems (with majority of power being generated by wind and solar-based power plants). The considered power systems comprise of conventional power plant modeled via a detailed nonlinear $9^{th}$-order dynamical model (modeling synchronous machine, exciter, and turbine dynamics), comprehensive power electronics based wind/solar power plants models, and composite load dynamics (constant power, constant impedance, and motor loads). The proposed WADCs act as a secondary control layer and are directly actuated via the primary controller layer (by sending additional control signals) of the power network.
	
	\item Two different approaches have been utilized to solve the same WADC problem for the considered power system. One is based on a completely model-free approach via reinforcement learning while the other methodology utilizes a model-based design. In particular, for the model-free approach, we leverage deep deterministic policy gradient (DDPG)-based algorithm to learn optimal control law by continuously interacting with the power system model. While for the model-based design, we use linear-matrix inequalities (LMIs)-based $\mathcal{H}_2$ stability notion along with Lyapunov control theory. 
	Also, to improve computational efficiency in the model-free approach, we use control-theoretic techniques to properly design and initialize the corresponding actor and critic neural networks as explained in Sec. \ref{sec:RL WAC}. 
	
	Thorough numerical simulations and discussions have been provided on modified IEEE 9-bus and 39-bus systems showcasing the different pros and cons of model-based and model-free WADC designs. Furthermore, to showcase the advantages of the proposed WADCs a comparative analysis has also been carried out by studying the transient response of the system with only primary controllers and with WADCs acting on top of them.
\end{itemize}

\noindent {\textbf{Notations:}}
Matrices and vectors are written in bold, while sets are represented using calligraphic fonts, such as $\mathcal{G}$ or $\mathcal{W}$. The notation $\mathbb{R}^{b}$ refers to the set of real-valued column vectors with \enquote*{b} elements, and $\mathbb{R}^{c \times d}$ denotes the set of real-valued matrices with \enquote*{c} rows and \enquote*{d} columns. The zero matrix is denoted by $\m{O}$, and the identity matrix of appropriate size by $\m{I}$. The union of two sets is represented by $\cup$, and the Kronecker product is indicated by $\otimes$. Additionally, $\mathbb{S}^{c \times d}_{++}$ refers to the set of positive definite matrices of size \enquote*{c} by \enquote*{d}. The asterisk $*$ in a symmetric matrix indicates that the entries are symmetric with respect to the main diagonal. All quantities are given in per unit (p.u.) unless otherwise noted. For simplicity, the time dependence of vectors is sometimes omitted in equations, e.g., $\m{x_d}(t)$ is written as $\m{x_d}$.

\vspace{-0.3cm}
\section{Renewables heavy power system with solar farms, wind farms, and composite loads}\label{sec: System model}
\vspace{-0.2cm}
We consider a power system model consisting of $G$ traditional power plants (both steam and hydro-based), $S$ solar power plants, $W$ wind power plants, and various loads: $L_k$ motor-based loads, $L_p$ constant impedance loads, and $L_z$ constant power loads. The electrical grid is represented as a graph, where $\mathcal{E}$ denotes the set of transmission lines and $\mathcal{N} = \left\lbrace 1, ..., N\right\rbrace$ represents the set of buses. These buses are categorized into different types: $\mathcal{G} = \left\lbrace 1, ..., G\right\rbrace$ corresponds to buses connected to traditional power plants, $\mathcal{S} = \left\lbrace 1, ..., S\right\rbrace$ to buses with solar farms, and $\mathcal{W} = \left\lbrace 1, ..., W\right\rbrace$ to buses connected to wind-based power plants. The set $\mathcal{L}$ represents buses connected to motor-based, constant impedance, and constant power loads, while $\mathcal{U}$ includes non-unit buses that are not connected to any generation or load elements.

With this setup, the power system is described by a set of nonlinear differential-algebraic equations (NDAEs) as follows\cite{nadeem2023RL}:
\begin{subequations}~\label{equ:PSModel}
	\begin{align}
		\dot{\m x}(t) &= \m f(\m x_d, \m x_a, \m w,  \m u) ~\label{equ:PSModel-a} \\
		\m 0 &= \m h(\m x_d,\m x_a, \m w,\m u). ~\label{equ:PSModel-b}
	\end{align}
\end{subequations}
In the above equations, differential equations \eqref{equ:PSModel-a} describe the models of traditional power plants, solar farms, wind power plants, and dynamics of composite loads while the algebraic equations \eqref{equ:PSModel-b} model the algebraic constraints (the power/current balance equations) in the network. The notation  $\m x_a(t) \in \mbb{R}^{n_a}$ represents algebraic variables, $\m x_d(t) \in \mbb{R}^{n_d}$ denotes dynamic variables, $\m u(t) \in \mbb{R}^{n_u}$ contains the control inputs, and $\m w(t) \in \mbb{R}^{n_w}$ denotes system disturbances. The detailed explanations of these vectors are given in \ref{appndix:ninth Gen_dynamics}. For brevity, the complete dynamical equations (set of ordinary differential equations)  describing the models of generators, solar, and wind power plants are not included in this paper and can be found in \cite{SoumyaITPWRS2022,pico2022blackstart,nadeem2023RL,sauer2017power}.

That being said, we can also express the power system model \eqref{equ:PSModel} in the following state-space format: 
\begin{align}\label{eq:final_NDAE}
\m E\dot{{\m x}} = \bar{\m A}{\m x}\hspace{-0.0cm}  + {{\bar{\m  B}}_u} {{\m u} } +\hspace{-0.0cm}  {\m f}\left({\m x},{\m w},{\m u} \right) + \bar{\m B}_w \m w
\end{align}
where $\m x(t) = \bmat{\m x_d^\top & \m x_a^\top}^\top \in\mbb{R}^{n}$ represents the overall state vector and $\m E \in\mbb{R}^{n\times n}$ is a singular matrix that encodes algebraic constraints, characterized by rows of zeros.   The constant matrices $\bar {\m A}\in\mbb{R}^{n\times n}$, $\bar{\m B}_u \in\mbb{R}^{n\times n_u}$, $\bar{\m B}_w \in\mbb{R}^{n\times n_w}$ map the state vector, control inputs, and the disturbance vector $\m w$ in the system dynamics. These matrices are determined by capturing the linear components of the model \eqref{equ:PSModel}  while the function ${\m f}\left({\m x},{\m w},{\m u} \right)$ accounts for any linearization errors present in the power system model.
\vspace{-0.5cm}
\section{Preliminaries and problem formulation}\label{sec:controller design}
\vspace{-0.2cm}
In this section, we address the state feedback control problem for the renewables-heavy interconnected power system model \eqref{eq:final_NDAE}. We first define the overall state feedback control problem and then propose two distinct solutions: one using a model-free approach via reinforcement learning, and the other using a model-based approach based on system matrix information and approximations of the nonlinear function ${\m f}\left({\m x},{\m w},{\m u} \right)$, employing Lyapunov stability theory.

In the designed system \eqref{eq:final_NDAE}, the dynamic states $\m{x_d}$—which represent conventional generators, renewable energy sources, and motor-based loads—are treated as dynamic components, while loads, non-unit buses, and other interconnections are considered static. For the purpose of designing a feedback controller, the algebraic state vector $\m{x_a}$, which includes voltage and current phasors, can be regarded as redundant and thus eliminated \cite{AranyaICSM2019}. Given that power system models are typically index-1 DAEs, these algebraic variables can be removed, allowing the conversion of the DAE into an equivalent ODE as follows:
\noindent Let us consider
\begin{align*}
	\bar{\m A} \hspace{-0.0cm}= \hspace{-0.0cm}\bmat{\m A_{dd}\;\; \m A_{da}\\ \m A_{ad} \;\; \m A_{aa}},\;  \bar{\m B}_u\hspace{-0.0cm} =\hspace{-0.0cm} \bmat{\m B_{ud}^\top\;\;\m B_{ua}^\top}\hspace{-0.0cm}^\top,\; \bar{\m B}_w \hspace{-0.0cm}=\hspace{-0.0cm} \bmat{\m B_{wd}^\top\;\;\m B_{wa}^\top}^\top \end{align*}
and assume that $\m A_{aa}$ is invertible (a common assumption in the literature of power systems). Then, we can extract the equation for $\m x_a$ and substitute it in the dynamic states equation to get the following ODE system:
\begin{align}\label{eq:final_NDAE}
	\dot{\m x}_d = {\m A}\m x_d\hspace{-0.0cm}  + {{{\m  B}}_u} {{\m u} } +\hspace{-0.0cm}  {\m f_d}\left(\m x_d,{\m w},{\m u} \right) + {\m B}_w \m w
\end{align}
where $\m f_d(.)$ represents the corresponding nonlinearity and the rest of the matrices are given as: 
\begin{align*}
	{\m A} &= \m A_{dd} - \m A_{da}\m A_{aa}^{-1}\m A_{ad},\;\;\;  	{\m B}_w = \m B_{wd} - \m A_{da}\m A_{aa}^{-1}\m B_{wa}\\
	{\m B_u} &= \m B_{ud} - \m A_{da}\m A_{aa}^{-1}\m B_{ua}.
\end{align*}


That being said, to set up the state feedback control problem in the control literature, one first needs to design the perturbed closed-loop system dynamics. For this purpose, for a dispatch time period $kT \leq t< (k+1)T$, let us consider a control policy given as:
\begin{align}\label{eq:control_policy}
	\boxed{\mr{\textbf{Control policy:}}\;\;\,{\m u}_{cl}(t) = \m u_{0}^k- \m K\left({\m x_d}(t) - \m x_d^k\right) }   
\end{align}
where  $\m x_d^k$ is the state equilibrium value before the occurrence of disturbances, $\m u_{{0}}^k$ is the set-point of the input $\m u$ which is determined for every $k^{th}$-dispatch time-period by running power flow (PF), and $\mK \in\mbb{R}^{n_u\times n_d} $ is a gain matrix (a design parameter). Then, the corresponding closed-loop system can be written as:
\begin{align}\label{eq:NDAE_cloosed_loop}
\m\dot{\m x}_d = {\m A}{\m x_d}\hspace{-0.0cm}  + {\m B_u} {{\m u_{cl}} } +\hspace{-0.0cm}  {\m f_d}\left({\m x_d},{\m u_{cl}},{\m w} \right) + {\m B}_w \m w
\end{align}

Now, let us assume there is an unknown disturbance in the power system. This disturbance will push the system to a new equilibrium state, let us denote that by $\m x_d'$. Then, the closed-loop system at this new equilibrium can be written as:
\begin{align}\label{eq:NDAE_x'}
	\m 0 &= \m A{\m x_d'}\hspace{-0.0cm} + \hspace{-0.0cm} \m B_u\m u_{cl}' \hspace{-0.0cm}+\hspace{-0.0cm}\m f_d\left(\m x_d',\m u_{cl}',\m w' \right)\hspace{-0.0cm}+\hspace{-0.0cm} \m B_w \m w'.
\end{align}
Then, the perturbed closed-loop dynamics can be expressed as follows (which is essentially computed  by subtracting \eqref{eq:NDAE_x'}  from \eqref{eq:NDAE_cloosed_loop}):
\begin{align}\label{eq:final_NDAE_peturbed}
\Delta\dot{\m x}_d\hspace{-0.0cm} &=\hspace{-0.0cm} (\m A\hspace{-0.0cm}-\hspace{-0.0cm}\m B_u\m K)\Delta\m x_d\hspace{-0.0cm}+\hspace{-0.0cm}\Delta\m f_d(\Delta \m x_d, \Delta\m u_{cl}, \Delta \m w)\hspace{-0.0cm}+\hspace{-0.0cm}\m B_w\Delta\m w \hspace{-0.0cm}\hspace{-0.0cm}
\end{align}
with $\Delta \m x_d = \m x_d-\m x_d'$, $\Delta \m w = \m w-\m w'$, $\Delta\m f_d(\Delta \m x_d, \Delta\m u_{cl}, \Delta \m w) = \m f_d(\m x_d,\m u_{cl},\m w)-\m f_d(\m x_d',\m u_{cl}',\m w')$. For simplicity, from now onward, with little abuse of notation, we drop the $\Delta$ notation from $\Delta \m x_d$,  $\Delta \m f_d$, and $\Delta \m w$ and simply use $\m f_d$, $\m x_d$, and $\m w$ instead, respectively. Thus the perturbed closed-loop system \eqref{eq:final_NDAE_peturbed} is  rewritten as follows:
\begin{align}\label{eq:final_NDAE_peturbed_final}
\dot{\m x}_d\hspace{-0.0cm} &=\hspace{-0.0cm} (\m A\hspace{-0.0cm}-\hspace{-0.0cm}\m B_u\m K)\m x_d\hspace{-0.0cm}+\hspace{-0.0cm}\m f_d( \m x_d, \m u_{cl},  \m w) \hspace{-0.0cm} + \hspace{-0.0cm}\m B_w\m w.\hspace{-0.0cm}
\end{align}

That being said, in state feedback control literature the primary goal is to design a control law (or control gain matrix $\m K$) with two main objectives. Firstly, the control law should converge the perturbed closed-loop system asymptotically to zero. This implies that the control policy \eqref{eq:control_policy} endeavors to restore the power system \eqref{eq:final_NDAE} to its steady-state equilibrium following a significant disturbance. Secondly, the designed control law should achieve this objective with the minimum effort required. This ensures that the control strategy minimizes the resources and inputs required to achieve the desired stabilization. We want to emphasize here that, in the end, the control policy $\m u_{cl}$ needs to be applied to the complete power system model \eqref{eq:final_NDAE}, then, it can simply be mapped back by computing $\hat{\m K} = \bmat{\m K & \mO} \in\mbb{R}^{n_u\times n}$  and redesigning the control policy as:  $\hat{\m u}_{cl} = \m u_{0}^k - \hat{\m K} \left({\m x}(t) - \m x^k\right)$. Now, since the dimensions are same, we can use  $\hat{\m u}_{cl}$ as a control policy for the complete power system NDAE.

Before introducing the state feedback control problem, we outline a key assumption used throughout the paper:
\begin{asmp}\label{asmp:regular}
	The pair $(\mA, \m B)$ is impulse controllable and finite dynamics stabilizable.
\end{asmp}
\noindent Assumption \ref{asmp:regular} is standard in control theory literature \cite{TakabaCDC,cobb1984controllability}, and several power system models have been shown to satisfy controllability and stabilizability, as demonstrated in \cite{AranyaICSM2019,Nugroho_ITCST2023}.

The overall infinite-horizon optimal state feedback control problem can then be formulated as:
\begin{align}\label{eq:initial_OP}
	\begin{split}
		\min_{{\m K} \in \mathbb{R}^{n_u \times n_d}} & \hspace{-0.0cm} \int_0^{\infty} ({\m C} \m x_d \hspace{-0.0cm}+\hspace{-0.0cm} {\m D} \m K\m x_d)^\top ({\m C} \m x_d \hspace{-0.0cm}+\hspace{-0.0cm} {\m D} \m K\m x_d)dt \\ \subjectto  & \;\;\;\mr{Dynamics}\; \;\eqref{eq:final_NDAE_peturbed_final}\;
	\end{split}
\end{align}
\noindent where $\mC \in \mathbb{R}^{n_d \times n_d}$ and $\mathbf{D} \in \mathbb{R}^{n_d \times n_u}$ are fixed penalty matrices, analogous to the weight matrices $\mQ$ and $\mR$ in traditional LQR control. The selection of $\mC$ and $\mD$ is driven by the grid operator's preferences, which determine the specific states or control inputs that should be penalized during the design of the controller gain matrix $\mK$.

To solve the above feedback control problem for the test system considered, we employ two distinct approaches. The first approach is a fully data-driven, RL-based technique, while the second is a model-based framework that leverages Lyapunov stability theory. Both approaches are discussed in detail in the following sections.
 \vspace{-0.3cm}
\section{Model-free optimal state feedback controller design}\label{sec:RL WAC}
\vspace{-0.03cm}
Here, we propose a model-free approach to solve the optimal state feedback control problem given in Eq. \eqref{eq:initial_OP}. The presented approach is based on deep reinforcement learning in which the agent learns optimal control policy (which essentially means learning the feedback controller gain matrix $\mK$) such that it maximizes cumulative rewards over time (the negative of quadratic cost function given in Eq. \eqref{eq:initial_OP}) by interacting with the environment (the complete NDAE power system model). 
That being said, since  \eqref{eq:initial_OP} is a \textit{continuous} action space infinite horizon problem, then deep-deterministic policy gradient based RL algorithms are well suited for such optimization problem \cite{DDPG}. An overview of DDPG-based RL is given below followed by the proposed methodology which properly shapes (using knowledge of state feedback control theory) the DDPG-based technique to efficiently solve the formulated state feedback control problem in \eqref{eq:initial_OP}.
\vspace{-0.2cm}
\subsection{DDPG for state feedback control in power systems}
\vspace{-0.0cm}
DDPG employs an actor-critic architecture, where the actor learns a policy for selecting actions, and the critic evaluates the value of those actions. Both the actor and critic functions are approximated using neural networks, let us denote them by $\m \pi_{\m \theta}(\m x_d):\mathbb{R}^{n_d}\rightarrow \mathbb{R}^{n_{u}}$ and $\m Q_{\m \phi}(\m x_d,\m u_\pi) :\mathbb{R}^{n_d}\times \mathbb{R}^{n_u}\rightarrow \mathbb{R}$, respectively, where $\m u_\pi$ is the output of the actor-network. The actor neural network is parameterized by $\m \theta$ and it takes states as inputs and returns action as output while for the critic network weights/parameters are denoted by  $\m \phi$ and it takes the pair states-actions as input and returns the Q-value (the cumulative long term reward for the state-action pair) as output.

That being said, generally speaking, the main idea in DDPG is to run gradient descent on the actor-network parameter $\m \theta$ as: $\m \theta \leftarrow \m \theta - \epsilon \nabla_{\m\theta}\m J(\m \theta)$, where $\epsilon$ represents the step size and  $\nabla_{\m\theta}\m J(\m \theta)$ denote the policy gradient which in DDPG is approximated as follows \cite{DDPG}:
\begin{align}\label{eq:actor_gardient}
\nabla_{\m\theta}\mJ(\m\theta) \approx \dfrac{1}{|\mathcal{B}|}\sum_{i \in \mathcal{B}}\nabla_u {\m Q}_\phi(.)\bigr|_{\m x_d[i], \m \pi_{\m \theta}(\m x_d[i])}\nabla_\theta \m \pi_{\m \theta}(\m x_d[i])
\end{align}
where $\m x_d[i], \m u[i]$ with $ i \in \mathcal{B}$ are set of samples extracted from the replay buffer which stores the history of observations  $(\m x_d, \m u, \m r, \m x_d^n)$, here, $\m r$ is the reward and  $\m x_d^n$ is the new state after taking action $\m u$. In \eqref{eq:actor_gardient} the critic network $\m Q_{\m \phi}(\m x_d,\m u)$ is learned via temporal difference learning given as \cite{DDPG}:
\begin{align}\label{eq:temporal_learn}
\min _\phi L(\m \phi) = \mathbb{E}[\m Q_{\m \phi}(\m x_d,\m u)\hspace{-0.0cm}-\hspace{-0.0cm}(\m r \hspace{-0.0cm}+ \hspace{-0.0cm}\m\gamma \m Q_{\m \phi_t}(\m x_d^n,  \m \pi_{\m \theta_t}(\m x_d^n)))^2]
\end{align}
where $\mathbb{E}$ denotes the expectation notation. In \eqref{eq:temporal_learn}  $\m \pi_{\m \theta_t}$, $\m Q_{\m \phi_t}$ represent the target actor and critic networks whose weights $\m \theta_t$,  $\m \phi_t$ are computed using  $\m \theta$,  $\m \phi$  via Polyak averaging as follows:
\begin{subequations}\label{eq:polyk_avg}
	\begin{align}
		\m \phi_t \leftarrow \rho \m \phi_{t}+(1-\rho) \m\phi\\
		\m \theta_t \leftarrow \rho \m \theta_{t}+(1-\rho) \m\theta
	\end{align}
\end{subequations}
where $\rho$ is a positive constant usually selected close to 1.  The target networks are essentially copies of the original networks which trail behind the original networks and are updated slowly using the original network parameters (in DDPG this is usually done to improve the training stability). Further detailed explanations about DDPG can be seen in \cite{DDPG}. 

\textcolor{black}{ With that in mind, to efficiently solve the state feedback control problem \eqref{eq:initial_OP} using the DDPG-based RL technique, we propose to do the following: Firstly, notice that, the actor neural network tries to approximate the feedback controller gain matrix $\mK$. For the perturbed closed-loop system, the control law implemented by the actor is essentially linear mapping $\m u_\pi= \m K\m x_d = k_1x_{d_1} + k_2x_{d_2}+...$ where $k_1$, $k_2+...$   are the weights of the actor neural network. Then, a shallow neural network for the actor has been designed with an input layer and a fully connected layer to provide a linear mapping between states to actions. Notice that to reduce the number of learnable parameters (and thus increasing the convergence rate) no \texttt{relu} layer is added as no nonlinear mapping is needed to train the actor.  Also, it is well known from the control theory that the closed-loop system is stable if these gains are negative, therefore, initializing them to take negative values can speed up the convergence. Moreover,  since the learned control law approximated by the actor neural network does not have any extra biasing (constant coefficients such as $\m u_\pi= \m K\m x_d + \m b$), then while training, the actor neural network weights are only updated and the biasing learn rate is set to be zero.} 

\textcolor{black}{Now, as discussed earlier, the critic learns the Q-value function. The critic accepts an observation-action pair as inputs and returns a scalar (the discounted long-term reward) as output.  To that end, from feedback control theory, we know that the long-term reward function (or the value function) for \eqref{eq:final_NDAE_peturbed_final} is known to be quadratic (the Lyapunov function $\m x_d^\top\m P\m x_d$ with $ \m P \succ 0$, which tells us the optimal cost-to-go). Thus, the critic network is designed to have a quadratic layer (which returns a vector of quadratic monomials) followed by a fully connected layer (providing a linear mapping of its inputs)}. Furthermore, in the context of our setting, we can essentially express the structure of the Q-value function as follows:
\begin{align}\label{eq:critc}
\m Q(\m x_d, \m u_\pi) \hspace{-0.0cm}= \hspace{-0.0cm}l_1x_{d_1}^2\hspace{-0.0cm}+\hspace{-0.0cm} l_2x_{d_1}x_{d_2}\hspace{-0.0cm}+\hspace{-0cm} l_3x_{d_2}^2\cdots l_{n_d}x_{d_1}u_1+\cdots
\end{align}
where $l_1, l_2, \cdots$ are the weights of the critic neural network, or alternatively Eq. \eqref{eq:critc} in matrix form can be written as:
\begin{align}\label{eq:critic2}
	\begin{split}
\m Q(\m x_d, \m u_\pi)\hspace{-0.0cm}  &=\hspace{-0.0cm} \bmat{x_{d_1}\;x_{d_2}\cdots u_1 \cdots}\hspace{-0cm}\bmat{l_1& \frac{l_2}{2}& \frac{l_4}{2}\cdots\\ \frac{l_2}{2}& l_3& \frac{l_5}{2}\cdots\\ \frac{l_4}{2}& \frac{l_5}{2}& l_6\cdots\\ \vdots &\vdots &\hspace{-0cm}\vdots }\hspace{-0cm}\bmat{x_{d_1}\\x_{d_2}\\ \vdots \\ u_1 \\ \vdots}
		\\
		\Longleftrightarrow & \bmat{\m x_d^\top& \m u_\pi^\top}\m L \bmat{\m x_d^\top& \m u_\pi^\top}^\top
	\end{split}
\end{align}
Since $\m u_\pi = \m K\m x_d$, then, one can rewrite \eqref{eq:critic2} as:
\begin{align}
\m Q(\m x_d, \m K\m x_d)\hspace{-0.0cm} &=\hspace{-0.0cm} \m x_d^\top\bmat{\mI\;\;\m K^\top}\m L \bmat{\mI\;\;\m K^\top}^\top\hspace{-0.0cm}\m x_d\hspace{-0.0cm} = \hspace{-0.0cm}\m x_d^\top\m P\m x_d
\end{align}
Now, in our case, since we are maximizing the negative of the quadratic cost function, both $\m P$ and $\m L$ should be negative definite. Thus, initializing the critic network to be a negative definite matrix can overall stabilize the learning and speed up the convergence.

Moving on, in DDPG for every episode, the actor generates a random action, applies it to the environment, and tries to maximize the cumulative long-term reward. However, for the power system model, any random unbounded action cannot be chosen as it will destabilize the system and the episode might not even start. Thus, for each episode, the random actions are selected in a bounded region with $\m u_{max}$ and $\m u_{min}$ as upper and lower bounds, respectively. For the considered test system the control actions are voltage and power/valve position set points as discussed in Sec. \ref{sec: System model}. Then, voltage setpoints are constrained between $\pm5\%$ while the power/valve position setpoints are selected to be between  $\pm10\%$. This is reasonable as in normal operation the voltage should be between $0.95-1.05$pu and also the power output from each power generator usually changes proportionally  in response to large disturbances \cite{AranyaICSM2019}. Furthermore, since the output of the actor-network is bounded between $\m u_{max}$ and $\m u_{min}$, then during training  to avoid the actor output being saturated frequently a \texttt{tanh} (which scales the output between $-1$ and $1$) and scaling layer (which scales back the output to their desired range) have been added. That being said, the overall proposed model-free approach to solving the state feedback control problem is given in Algorithm \ref{alg:Algorithm 1}.
\begin{algorithm}[h]\label{alg:Algorithm 1}
	\caption{\text{Model-free state feedback controller}}\label{alg:model-free RL}
	\DontPrintSemicolon 
Design actor $\m \pi_{\m \theta}$  and critic $\mQ_\phi$ networks and initialize them as discussed in Sec \ref{sec:RL WAC}.\;	
Initialize $\rho$ and target actor $\m \pi_{\m \theta_t}$ and critic $\mQ_{\phi_t}$ networks.\; 
Initialize $\m C$,  $\m D$ matrices and control input bounds $\m u_{max}$, $\m u_{min}$.\;
Initialize replay buffer $\mathcal{R}$.\;
\For{{{episode} = 1 to M}}{
Initialize power system model with a random initial condition $\m x_0$ having $5\%$ maximum deviation from steady-state values.\;
\For{{{iteration} = 1 to J}}{
	Select input $\m u_\pi$ from the actor network.\;
	Execute $\m u_\pi$ on system \eqref{eq:final_NDAE_peturbed_final} using $\m x_0$ and observe quadratic cost $\m r$ and next state $\m x_d^n$.\;
	Store observation in $(\m x_d, \m u_\pi, \m r, \m x_d^n)$ in  $\mathcal{R}$.\;
	Sample mini-batch $\mathcal{B}$ of observation from $\mathcal{R}$.\;
	\For{observation $ b \in \mathcal{B}$ }{ 
	Compute: $\m y_b = \m r(\m x_d, \m u)+ \m\gamma \m Q_{\m \phi_t}(\m x_d^n,  \m \pi_{\m \theta_t}(\m x_d^n))$.\;}
	Update the critic network by minimizing the loss as given in \eqref{eq:temporal_learn}.\;
	Update the actor network via gradient descent and with policy gradient  $\nabla\m J(\m \theta)$  computed using Eq.  \eqref{eq:actor_gardient}.\;
	Update target actor and crtic network through polyak averaging using Eq. \eqref{eq:polyk_avg}.\;
}}
Return trained actor $\m \pi_{\m \theta}$  and critic $\mQ_\phi$ networks.\;
Extract controller gain $\m K$ from $\m \pi_{\m \theta}$.\;
Design $\hat{\m K}$, plug it in the control policy $\hat{\m u}_{cl}$, and apply it to the complete NDAE power system \eqref{eq:final_NDAE} with $\m x^k$ and $\m u_0^k$ being updated every OPF/PF dispatch time period $kT \leq t< (k+1)T$. 
\end{algorithm}
\vspace{-0.3cm}
\section{Model-based optimal state feedback controller design}\label{sec:Model-based WAC}
\vspace{-0.1cm}
To solve the formulated optimal state feedback control problem \eqref{eq:initial_OP}, here, we develop a completely model-based approach (without relying on any data or interaction with the power system environment). The proposed technique in this section is based on the Lyapunov stability theory and it consider information of the constant system matrices $(\mA, \mB_u, \mB_w)$ along with approximation on nonlinear function ${\m f_d}\left({\m x_d},{\m u},{\m w} \right)$ to compute the controller gain matrix $\mK$.

That being said, to solve the optimal feedback control problem \eqref{eq:initial_OP}, we present the following results:
\begin{myprs}\label{theorm:H_inf}
	Suppose Assumption \ref{asmp:regular} hold. Then, there exists a solution to problem \eqref{eq:initial_OP} (meaning the perturbed closed-loop dynamics asymptotically converges to zero with minimum control effort required), if there exist matrices $\m Z \in \mbb{S}_{++}^{n_d \times n_d}$, $\m W \in \mbb{R}^{n_d \times n_d}$, and $\m F \in \mbb{R}^{n_u \times n_d}$ that are solution to the following semi-definite program (SDP)  
	\begin{align*}
 		\mathbf{\left( SDP\right) }\;\;\;\;\; \min_{\m F, \m Z,\m W} \;\;\;  Tr (\mW)\\ \subjectto  & \;\;\;\mr{LMI}\; \eqref{eq:LMI},\; \; \m Z \succ \mO \bmat{\mZ & *\\ \hat{\m B}_w^\top &\mW}\succeq 0
	\end{align*} 
	where $\mr{LMI}$ \eqref{eq:LMI} is as follows:
	\begin{align}\label{eq:LMI}
\bmat{\mZ^\top\mA^\top \hspace{-0.0cm}+\hspace{-0.0cm} \mA\mZ\hspace{-0.0cm} -\hspace{-0.0cm} \mF^\top\mB_u^\top\hspace{-0.0cm}-\hspace{-0.0cm}\mB_u\mF & *  & * \\  \hat{\m B}_w^\top& \mO&*\\ \mC\m Z + \mD\mF& \mO&-\mI}\prec \mO
	\end{align}
Upon solving the above SDP the controller policy can be computed as $\m K = \mF\mZ^{-1}$.
\end{myprs}
\begin{proof}
Before presenting the proof, to formulate a tractable convex LMI-based formulation for the model-based controller design, we assume that the perturbation in the nonlinearity in Eq. \eqref{eq:final_NDAE_peturbed_final} is $\mathcal{L}_2$-norm bounded and can be expressed as $\Delta\m f_d(\m x_d,\m u_{cl},\m w) = \mB_{fd}\m w_{fd}$ with $\m B_{fd} = \m B_w$. We want to emphasize here that, the above assumption on $\Delta\m f_d(\m x_d,\m u_{cl},\m w)$ is only carried out to design tractable convex SDP formulation for the controller design. At the end, the final designed feedback controller is applied to the complete NDAE power network without any simplifications. Now, we define:
	\begin{align}\label{eq:B_w tilde}
		\hat{\m w}\hspace{-0.0cm} =\hspace{-0.0cm} \bmat{\m w^\top & \m w_{fd}^\top}^\top,~ \hat{\mB}_w \hspace{-0.0cm}=\hspace{-0.0cm} \bmat{\mB_w & \m B_{fd}}.
	\end{align}
With these definitions in place, to begin the proof of Proposition \ref{theorm:H_inf},  we consider a candidate Lyapunov function for the perturbed system \eqref{eq:final_NDAE_peturbed_final} as $V(\m x_d)=\m x_d^\top \mP \m x_d $ where  $V:\mbb{R}^{n_d}\rightarrow \mbb{R}_+$ and $\m P\in\mbb{R}^{n_d\times n_d}$. Then, the derivative of $V(\m x_d)$ with respect to $\m x_d$ along the trajectories of \eqref{eq:final_NDAE_peturbed_final} can be expressed as
	\begin{align*}
		{\dot V}(\m x_d) &= (\mP \m x_d)^\top\dot{\m x_d}+ \dot{\m x_d}^\top (\mP \m x_d).
	\end{align*}
	Now, for the overall system stability including the objective function (as given in \eqref{eq:initial_OP}), we need ${\dot V}(\m x_d) + \m x_d^\top \m\Gamma^\top \m \Gamma\m x_d < 0$ where $\m \Gamma = \m C - \m D\mK$, which can expanded as:
	\begin{align*}
		 {(\m P\m x_d)}^\top(\m{A}_{C}\m x_d+ \hat{\m B}_w \hat{\m w})  + \m x_d^\top \m\Gamma^\top \m \Gamma\m x_d + ( \m{A}_{C}\m x_d + \hat{\m B}_w \hat{\m w})^\top\m{P}\m x_d< 0
	\end{align*}
where $ \m{A}_{C} = \m A-\mB_u\mK$. We can rewrite the above equation also as $\m\Psi^\top\m\Xi\m\Psi<0$ with:
	\begin{align}\label{eq:proof_2}
		 \m\Psi\hspace{-0.01cm} = \hspace{-0.01cm}\bmat{\m x \\ \hat{\m w}}^\top, \m\Xi\hspace{-0cm}=\hspace{-0cm}\hspace{-0cm}\bmat{ \mA^\top_{C}\mP\hspace{-0.01cm}+\hspace{-0.01cm}\mP^\top \mA_{C} \hspace{-0.01cm}+\hspace{-0.01cm}  \m\Gamma^\top \m \Gamma & \mP^\top\hat{\m B}_w \\ \hat{\m B}_w^\top\mP & \mO}
	\end{align}
Notice, $\m\Psi^\top\m\Xi\m\Psi<0$ holds if and only if $\m\Xi\preceq \mO$.

Now, let us define  $\mZ:=\mP^{-1}$ and pre-multiply and post-multiply \eqref{eq:proof_2} with   $\mr{diag}\left( \bmat{\mZ^\top&\mI}\right)$ and $\mr{diag}\left( \bmat{\mZ&\mI}\right)$, respectively, to get the following equivalent representation of \eqref{eq:proof_2}:
\begin{align*}
	\hspace{-0.0cm}\bmat{\mZ^\top\mA^\top_{C}+ \mA_{C}\mZ+\mZ^\top\m\Gamma^\top \m \Gamma\m Z\hspace{-0.01cm} & \hspace{-0.0cm}\hat{\m B}_w  \\  \hat{\m B}_w^\top& \mO}\prec \mO
\end{align*}
Then, applying the Schur complement lemma \cite{zhang2006schur}, the above LMI can equivalently be represented as:
\begin{align}\label{eq:prf_cong1}
	\hspace{-0.0cm}\bmat{\mZ^\top\mA^\top_{C}+ \mA_{C}\mZ & \hat{\m B}_w  & \mZ^\top\m\Gamma^\top\\  \hat{\m B}_w^\top& \mO&\mO\\ \m\Gamma\m Z& \mO&-\mI}\prec \mO
\end{align}
Now, defining $\mF := \m K\m Z$, then we get the convex LMI \eqref{eq:LMI} which represents the necessary and sufficient stability conditions for the existence of the controller policy. Finally, to minimize the impact of $\m w$ on the system dynamics according to $\mathcal{H}_2$ notion, we need to minimize $Tr(\hat{\m B}_w^\top \mZ \hat{\m B}_w)$ where $Tr$ denotes the Trace notation. Now, let us upperbound $\hat{\m B}_w^\top\m Z \hat{\m B}_w$ it by matrix variable $\m W$ as $\hat{\m B}_w^\top \mZ \hat{\m B}_w \preceq \m W$. Then, by minimizing  $Tr(\m W)$ and taking Schur compliment of $\hat{\m B}_w^\top\mZ \hat{\m B}_w < \m W$ we get the second LMI along with the objective function in the formulated SDP. This completes the proof. 
\end{proof}
 Proposition \ref{theorm:H_inf} presents a completely model-based approach to solve the formulated optimal state feedback control problem  \eqref{eq:initial_OP}. The computed gain matrix $\mK$ guarantees the asymptotic stability of the perturbed system \eqref{eq:final_NDAE_peturbed_final}. In other words, it ensures that after a large disturbance, the NDAE system model \eqref{eq:NDAE_cloosed_loop} converges back to its steady values (while adding damping to the system oscillation) with minimum control effort required. 

We want to emphasize here that, in the end, both proposed approaches provide us the feedback gain matrix $\m K$ via solving exactly the same state feedback control problem given in \eqref{eq:initial_OP}. In the model-free approach, the controller tries to learn $\mK$ by continuously interacting with the power system model and observing the quadratic cost/reward given in \eqref{eq:initial_OP} while in the model-base design, it requires an accurate system model to compute $\mK$. As renewable energy integration and the widespread adoption of distributed energy resources continue to evolve, the accurate modeling of power systems becomes increasingly challenging. Thus, the model-free approach has a major advantage over model-based design as it does not require any system information and modeling.

In the following section, we implement both presented approaches and do a thorough comparison on a renewables heavy power system model with detailed synchronous generator, wind, solar, and composite loads (constant power, constant impedance, and motor-based loads) dynamics as discussed in Sec. \ref{sec: System model}. 


\vspace{-0.4cm}
\section{Case studies}\label{sec:case studies}
\vspace{-0.2cm}
Both the proposed model-based and model-free controllers have been tested on  IEEE 9-bus and 39-bus systems. Both these power systems have been modified to include composite loads, and high penetration of wind, solar-based renewable resources. The 9-bus system includes one steam-based traditional power plant at Bus 1, one wind-based power generator at Bus 3, a solar farm at Bus 2, and a motor-based load at Bus 8. Similarly, the 39-bus system is comprised of four traditional power plants, four solar power plants, two wind-based power plants, and a motor-based load at Bus 14. The one-line diagrams of both power systems are presented in \ref{appndix:ninth Gen_dynamics}, Figs. \ref{fig:Diagram39}, \ref{fig:Diagram}. Further details about modeling and parameters of traditional power plants, wind/solar farms, and loads of the systems used in this work are given in \cite{sauer2017power,SoumyaITPWRS2022, Hart2016, WasynczukITPE1996}.

All the case studies are performed on a computer with Intel $i7$ processor and $64$GB RAM. The power systems have been modeled in MATLAB 2021a and are simulated using \texttt{ode15s} MATLAB DAE solver. The initial conditions for the power systems are computed using power flow studies via MATPOWER \cite{Matpower} \texttt{runpf} function. The baseline for the frequency and volt-ampere of the power network is set to be  $w_b = 120\pi\mr{rad/s}$, $S_b = 100\mr{MVA}$, respectively. 
Also, it's worth mentioning that in both the 9-bus and 39-bus systems the total shares of power production from renewables are around $63\%$ and $67\%$, respectively.

To design the actor and critic networks and to implement the proposed DDPG-based model-free feedback controller, MATLAB RL-toolbox is utilized. During learning for the 9-bus system the agent interacts with the power system model for a sample time of $T_s = 0.1s$ and a total of $200$ samples are taken in each episode while the total number of episodes is set to be $4000$. The rest of the hyper-parameters for both 9-bus and 39-bus systems used during the learning process are given in Tab. \ref{Tab: table para}. 
These hyper-parameters are essential and need to be properly tuned to have stable training and to get optimal control policy. Although there is no systematic way to tune these parameters, they can be determined based on human operator's knowledge or trials and errors. One can play with these parameters to get better performance.  In Tab. \ref{Tab: table para}, we set the critic learning rate to be larger as compared to the actor learning rate. This seems to work well in our case, the intuitive reasoning behind this can be that since the critic takes the action and observation and tells how good the taken action is, thus pushing the critic to learn quickly as compared to the actor seems to stabilize the overall learning.

In every episode, the initial conditions for the system are randomly chosen with $5\%$ maximum deviation from steady-state values. Notice that, while learning in every episode, if for any particular state-action pair the \texttt{ode15s} solver does not converge to a solution (meaning the power system is not simulated) then no learning is performed and the episode is terminated. Regarding training, the overall learning times for both test systems are given in Tab. \ref{tab:Table 1} while the cumulative long-term rewards (the negative of the quadratic cost function) for both test networks are presented in Fig. \ref{fig:pd rewards}. We can see that for the 9-bus test system the algorithm converges to optimal policy in around 4000 episodes while for the 39-bus test system, it took almost 5000 episodes to converge and learn the optimal control policy.
\begin{figure}
	\hspace{-0.3cm}\subfloat{\includegraphics[keepaspectratio=true,scale=0.53]{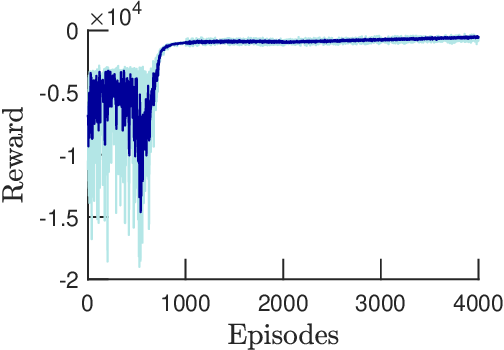}}{}\hspace{-0.1cm}
	\subfloat{\includegraphics[keepaspectratio=true,scale=0.53]{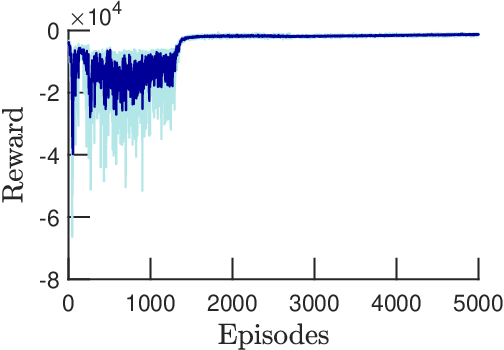}}{}{}\hspace{-0.25cm}
	{}{}\vspace{-0.3cm} \caption{Average and episodic rewards for 9-bus (left) and 39-bus (right) test systems.}\label{fig:pd rewards}
 \vspace{-0.65cm}
\end{figure}

On the other hand, the model-based controller is implemented in YALMIP \cite{Lofberg2004} with MOSEK \cite{Andersen2000} as an optimization solver to solve the proposed SDP given in Proposition  \ref{theorm:H_inf}. Notice that, in both types of designs, the proposed controller acts in realtime as a secondary control loop and is directly actuated via the primary controllers (using their voltage and power/valve-position set points).


\begin{figure}
	\subfloat{\includegraphics[keepaspectratio=true,scale=0.53]{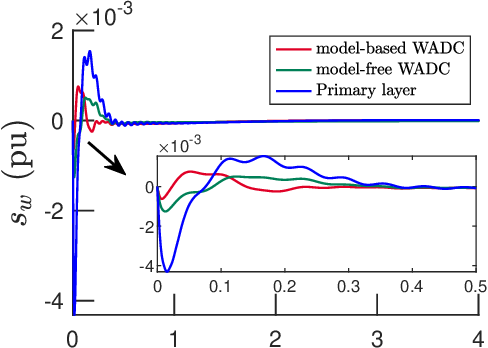}}{}\hspace{-0.05cm}
	\subfloat{\includegraphics[keepaspectratio=true,scale=0.53]{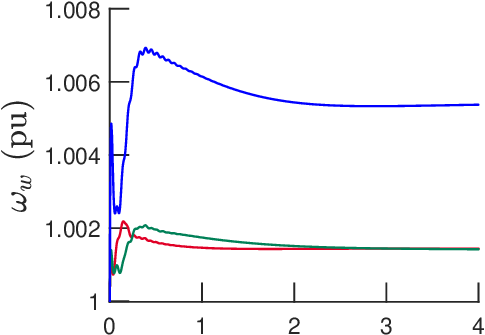}}{}{}\hspace{-0.25cm}
	\subfloat{\includegraphics[keepaspectratio=true,scale=0.53]{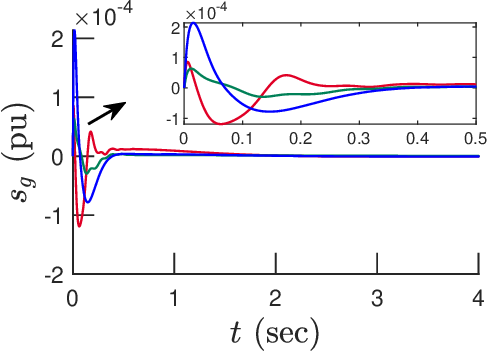}}{}{}\vspace{0.19cm}\hspace{-0.05cm}
	\subfloat{\includegraphics[keepaspectratio=true,scale=0.53]{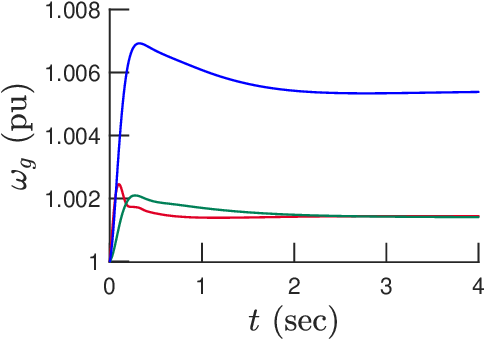
	}}{}{}\vspace{-0.5cm} \caption{Comparative analysis under $\Delta_L = -0.5$ for 9-bus system: relative slip and angular speed of wind power plant (above),  slip and rotor frequency of synchronous machine (below).}\label{fig:pd decrease case 9}
 \vspace{-0.65cm}
\end{figure}

\begin{figure}
\subfloat{\includegraphics[keepaspectratio=true,scale=0.53]{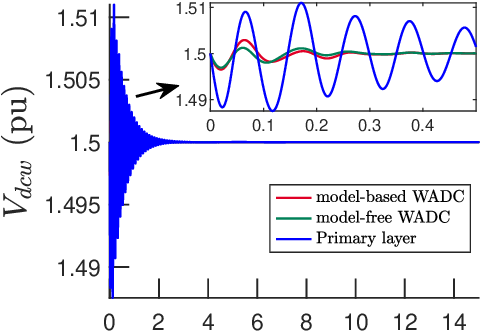}}{}\hspace{-0.05cm}
	\subfloat{\includegraphics[keepaspectratio=true,scale=0.53]{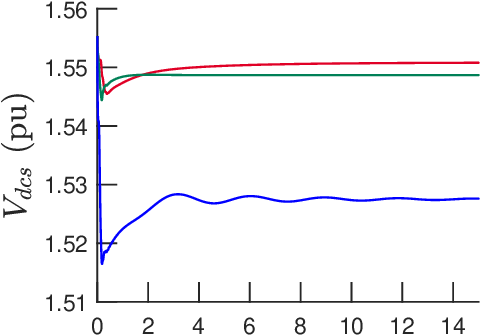}}{}{}\hspace{-0.25cm}
	\subfloat{\includegraphics[keepaspectratio=true,scale=0.53]{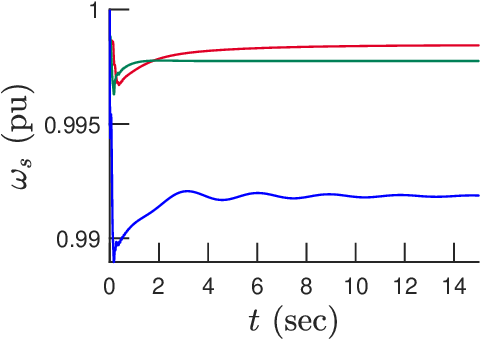}}{}{}\vspace{0.19cm}\hspace{-0.05cm}
	\subfloat{\includegraphics[keepaspectratio=true,scale=0.53]{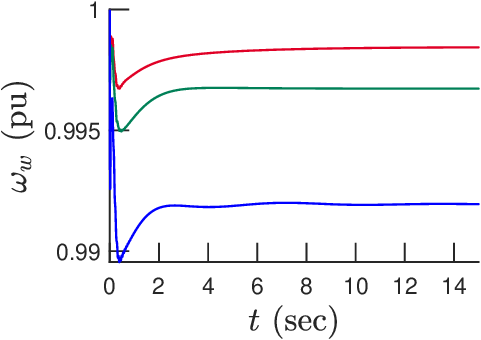
	}}{}{}\vspace{-0.5cm} \caption{Comparative analysis under $\Delta_L = 0.7$ for 9-bus system: DC link voltage of the wind and solar plant (above), relative speed of wind and solar plants (below).}\label{fig:pd inc case 9}
 \vspace{-0.5cm}
\end{figure}

\begin{figure}[h]
	\subfloat{\includegraphics[keepaspectratio=true,scale=0.53]{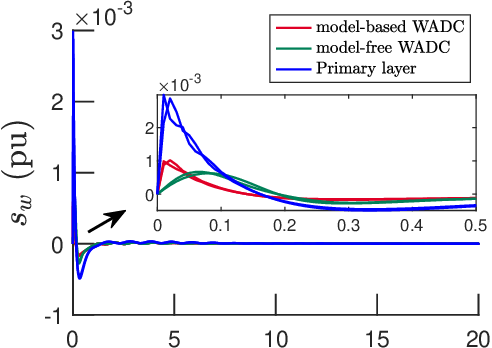}}{}\hspace{-0.05cm}
	\subfloat{\includegraphics[keepaspectratio=true,scale=0.53]{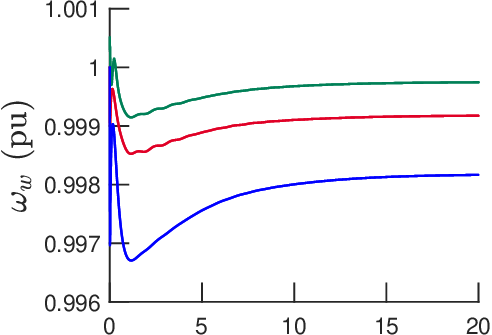}}{}{}
	
	\subfloat{\includegraphics[keepaspectratio=true,scale=0.53]{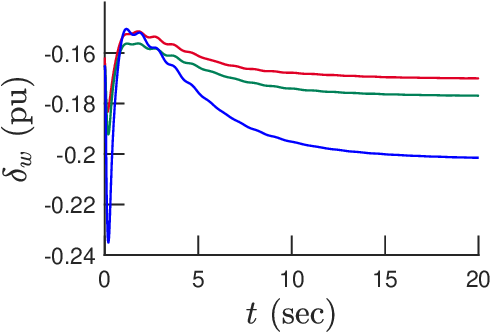}}{}{}\hspace{-0.05cm}
	\subfloat{\includegraphics[keepaspectratio=true,scale=0.53]{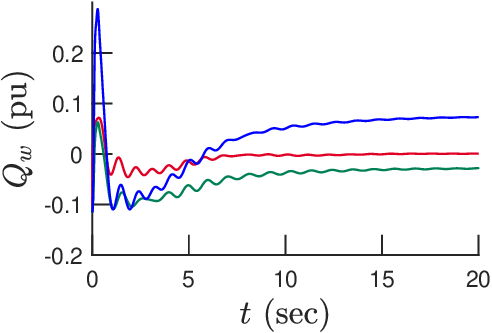
	}}{}{}\vspace{-0.69cm} \caption{Comparative analysis under $\Delta_L = 0.003$ for 39-bus test system: relative slip and angular speed of all wind power plants (above),  relative angle and reactive power output of wind plant connected at Bus 32 (below).}\label{fig:pd increase case 39}
 \vspace{-0.5cm}
\end{figure}

\begin{figure}
	\subfloat{\includegraphics[keepaspectratio=true,scale=0.53]{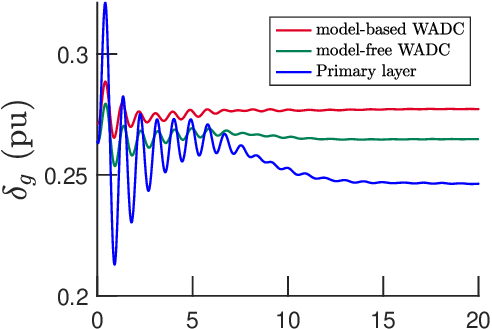}}{}\hspace{-0.05cm}
	\subfloat{\includegraphics[keepaspectratio=true,scale=0.53]{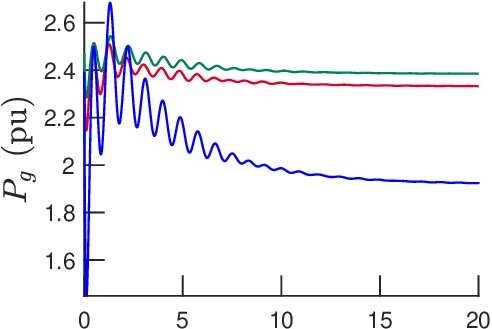}}{}{}
	
	\subfloat{\includegraphics[keepaspectratio=true,scale=0.53]{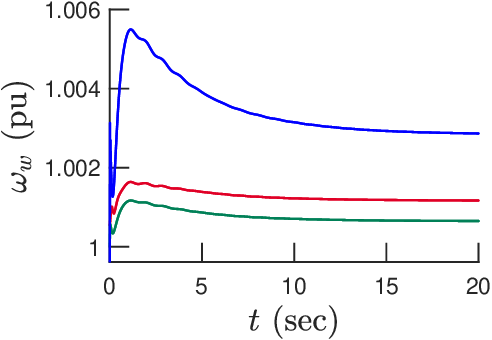}}{}{}\hspace{-0.05cm}
	\subfloat{\includegraphics[keepaspectratio=true,scale=0.53]{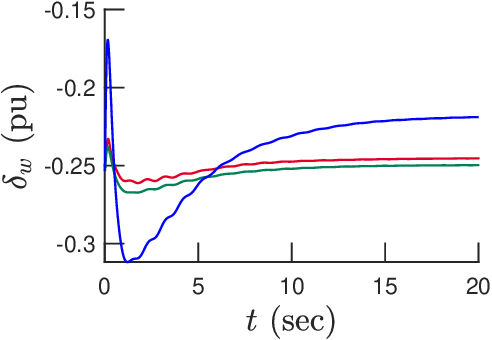}}{}{}\vspace{-0.69cm} \caption{Comparative analysis under $\Delta_L = -0.001$ and $\Delta_{I_s} = 0.1$ for 39-bus test system: rotor angle and real power output of Generator at Bus 38 (above),  relative speed of all wind power plants and relative angle of wind plant connected at Bus 34 (below).}\label{fig:pd dec case 39}
 \vspace{-0.5cm}
\end{figure}

\begin{table}[t!]
	\caption{Parameters used in model-free state feedback controller design.}\label{Tab: table para}
		\setlength{\tabcolsep}{1.5pt}
	\small
	\centering
	\begin{tabular}{|c|c|c|}
		\hline
		{Parameter} & {Value, 9-bus} & {Value, 39-bus} \\
		\hline
		\hline
		\hspace{-0.1cm}Sample time, $T_s$ & \hspace{-0.1cm} $0.1s$ & $0.3s$ \\
		\hline
		\hspace{-0.1cm} Replay buffer length, $\mathcal{R}$ & \hspace{-0.1cm} $10^6$ & $10^8$ \\
		\hline
		\hspace{-0.1cm} Mini batch size, $\mathcal{B}$ & \hspace{-0.1cm} $10^1$ & $10^2$\\
		\hline
		\hspace{-0.1cm} actor learning rate & \hspace{-0.1cm} $10^{-4}$ & $10^{-4}$\\
		\hline
		\hspace{-0.1cm} critic learning rate & \hspace{-0.1cm} $10^{-1}$ & $10^{-1}$\\
		\hline
		\hspace{-0.1cm} Max episodes, M & \hspace{-0.1cm} 4000 & 5000 \\
		\hline
		\hspace{-0.1cm} Max steps per episode & \hspace{-0.1cm} 200 & 300\\
		\hline
	\end{tabular}
	\vspace{-0.5cm}
\end{table}


\vspace{-0.2cm}
\subsection{Comparative analysis under uncertainty in load demand}
\vspace{-0.02cm}
In this section, we do a thorough comparative analysis of the proposed model-based and model-free state feedback controllers under abrupt disturbances in load demand. To that end, we carried out the numerical simulations as follows:
Initially, the power network operates under steady-state conditions, meaning the power generation is equal to the overall load demand and thus the system rests at an equilibrium. Then, right at the start of simulation, at $t>0$, an abrupt change in overall load demand occurs as follows:
$P^n_L =(1+\Delta_L)P^0_L$, where $\Delta_L$ denotes the severity of the load disturbance and $P^0_L$ represents the initial load demand before the disturbance, which, for the considered test systems are; $0.77$pu for the 9-bus test system and $19.8$pu for 39-bus test network. With that in mind, two different simulation studies are conducted for both test systems. In the first study, we assume  $\Delta_L$ to be positive---meaning the overall load demand has suddenly been increased, which in other words can roughly be interpreted as a generator-trip event. That being said, for the 9-bus test system, we select  $\Delta_L = 0.7$ while for 39-bus network we select it to be $\Delta_L = 0.003$. For the second simulation, we assume $\Delta_L$ to be negative, meaning the overall load demand of the power network has suddenly been decreased, thus it can be seen as a load-trip event. With that in mind, for this simulation, we set  $\Delta_L$ for the 9-bus and 39-bus test systems to be  $-0.5$ and  $-0.001$, respectively.

The above disturbance is going to push the power system to a new equilibrium from its initial steady-state or it might potentially destabilize the whole system. The objective of the proposed WADCs is to hedge against this disturbance and enhance system transient stability by damping LFOs/ULFOs, improve frequency nadir, and bring the system back to its nominal operating conditions. 
The results are presented in Figs. \ref{fig:pd decrease case 9}, \ref{fig:pd inc case 9}, and \ref{fig:pd increase case 39}. To demonstrate the effectiveness of the proposed WADCs, a comparison of the system transient response with only primary control layer and with developed WADCs acting on top of them has also been presented. Note that, by primary control layer, we refer to the conventional control mechanism of the power system, which in the considered test system for synchronous machines are;  power system stabilizers (PSSs), governor, automatic voltage regulators (AVRs), and machine inertial response. Similarly, wind and solar power plants are acting in grid-forming mode with droop controllers and PI-type current and voltage regulators acting as their control mechanisms. Providing detailed descriptions and mathematical equations of synchronous machine, solar/wind models and their primary controllers is out of the scope of this study and readers are referred to \cite{sauer2017power, DudgeonITPWRS2007} for
synchronous machines and \cite{SoumyaITPWRS2022, WasynczukITPE1996, pico2014voltage} for wind and solar power plants.

To that end, for the 9-bus test system, we can see that with both the proposed damping controllers there are significant improvements in LFOs/ULFOs, this can be verified by looking at the plots of the synchronous machine and solar/wind slips given in Figs. \ref{fig:pd decrease case 9} and \ref{fig:pd inc case 9}.  We can see that with only the primary control layer right after the load disturbance during transient period (the initial 2-3 seconds of the simulations) there are significant oscillations, while with the additional proposed secondary control layer, the oscillations have been damped out. Similar results have been achieved for the 39-bus test system as given in Fig. \ref{fig:pd increase case 39}. Thus, with the proposed WADCs the overall power system transient stability has been improved.

{\textcolor{black}{To further highlight the advantages of our feedback controllers, we also evaluated their performance under different initial power flow conditions and a line-to-ground fault on the 9-bus test system. In the former case, we varied the system loading, ran a power flow analysis to obtain the algebraic variables, and then introduced a step disturbance at the start of the simulation to create a transient response. In the latter scenario, we applied a line-to-ground fault on the transmission line between buses $4$ and $6$ at $t = 4$sec clearing it at $50$ ms from the near end and $200$ ms from the remote end. As shown in Figs. \ref{fig:pd dec case 9_pf} and \ref{fig:pd dec case 9_fau}, both proposed controllers effectively damp the system-wide oscillations, thus improving transient stability under these diverse operating conditions.}}

\begin{figure}
	\hspace{-0.25cm}\subfloat{\includegraphics[keepaspectratio=true,scale=0.53]{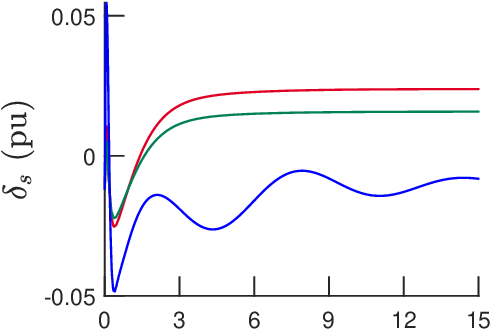}}{}\hspace{0.25cm}
	\subfloat{\includegraphics[keepaspectratio=true,scale=0.53]{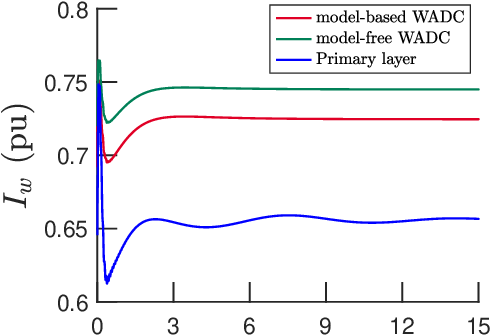}}{}{}
	
	\subfloat{\includegraphics[keepaspectratio=true,scale=0.53]{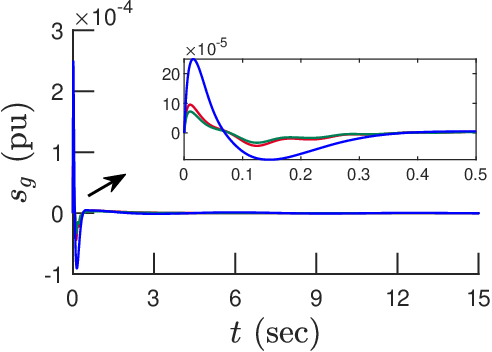}}{}{}\hspace{-0.05cm}
	\subfloat{\includegraphics[keepaspectratio=true,scale=0.53]{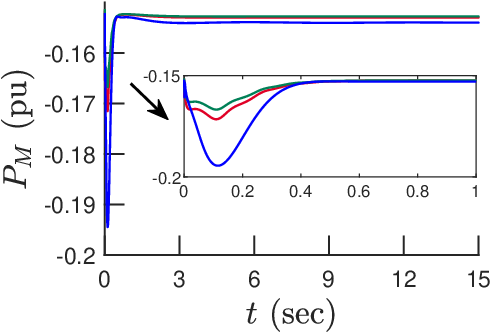}}{}{}\vspace{-0.69cm} \caption{\textcolor{black}{Comparative analysis under different initial power flow loading conditions and step disturbance for 9-bus test system: inverter angle of solar plant (top left), wind power plant  current output (top right), generator slip (bottom left), and motor power (bottom right).}}\label{fig:pd dec case 9_pf}
	\vspace{-0.5cm}
\end{figure}
\begin{figure}
	\hspace{-0.1cm}\subfloat{\includegraphics[keepaspectratio=true,scale=0.53]{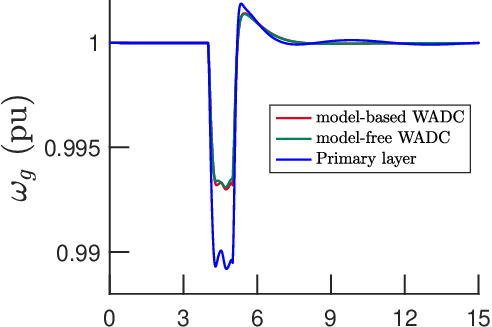}}{}\hspace{-0.05cm}
	\subfloat{\includegraphics[keepaspectratio=true,scale=0.53]{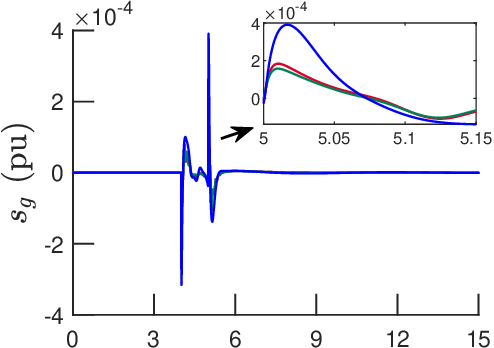}}{}{}
	
	\subfloat{\includegraphics[keepaspectratio=true,scale=0.53]{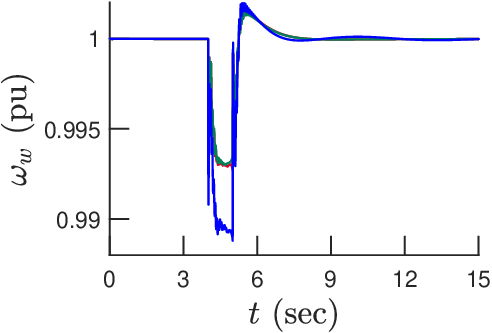}}{}{}\hspace{-0.05cm}
	\subfloat{\includegraphics[keepaspectratio=true,scale=0.53]{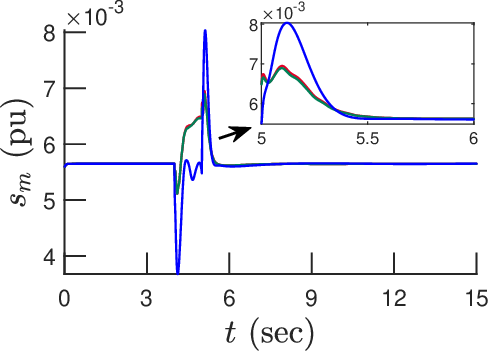}}{}{}\vspace{-0.69cm} \caption{\textcolor{black}{Comparative analysis under line-to-ground fault for 9-bus test system: generator speed and slip (above), angular speed of wind power plant (bottom left), and motor slip (bottom right).}}\label{fig:pd dec case 9_fau}
	\vspace{-0.5cm}
\end{figure}
\vspace{-0.24cm}
\subsection{Comparative analysis  under renewable uncertainty}
\vspace{-0.01cm}
Here we thoroughly study the transient behavior of the considered test system with and without the proposed damping controllers under uncertainty in the renewables. To that end, the simulations here are carried out as follows: At the beginning, the solar irradiance on all solar power plants  is set to be $1000$ $W/m^2$ (which is the standard solar irradiance), then, at $t>0$ we decrease the irradiance of PV plant connected at Bus 1 for 9-bus system and PV plants at Buses 30, 37 for 39-bus test system as follows:
$I^n_s =(1-\Delta_{I_s})I^0_s$, where $I^0_s$ is the solar irradiance before the disturbance and $I^n_s$ represent its new value, while $\Delta_{I_s}$ represent  the severity of disturbance and is set to be $0.1$ (meaning the irradiance is decreased by $10\%$).

To further add transients, the load disturbance from the previous section has also been kept intact and to further mimic realistic load uncertainty, we add some Gaussian noise to it also as: $P^n_L =(1+\Delta_L)P^0_L + \omega_L(t)$, where $\omega_L(t)$ denotes Gaussian noise with zero mean and standard deviation of $0.2\Delta_L$. Under these transient conditions, the system is initially stabilized with only the primary control layer and then the proposed damping controllers are also added on top of them.

The results are presented in Figs. \ref{fig:pd dec case 39} and \ref{fig:pd dec case2 39}. 
We can see that with both the proposed WADCs there is an improvement in system oscillations. This can be verified from the plots of the slips of all the power plants. Similarly, from Fig.  \ref{fig:pd dec case 39}  the oscillations in the power output of synchronous machines have also been damped out during transient periods. 
Note that, in Figs. \ref{fig:pd dec case 39} and \ref{fig:pd dec case2 39}, we compute the slip $\m s$, inverters relative angular speed $\m \omega$, and  DC-link voltages $\m V_{dc}$ from the state vector using the following equations \cite{SoumyaITPWRS2022}:
\begin{align*}
	\m\omega(t)\hspace{-0.09cm} = \hspace{-0.09cm} 1\hspace{-0.09cm}-\hspace{-0.0cm}k_d({\mP(t)}\hspace{-0.09cm}-\hspace{-0.09cm}\mP^*),\;\;
	\m V_{dc}(t) \hspace{-0.09cm}=\hspace{-0.09cm} \sqrt{\m E_{dc}(t)},\;\; \m s(t) \hspace{-0.09cm} =\hspace{-0.09cm} (w_c\hspace{-0.09cm} -\hspace{-0.09cm} \m\omega(t))/w_c
\end{align*}
where $k_d$ is the droop constant of the power  plant,  $\m P$ is real power output of the power plant, while $w_c$ is the overall weighted-average system frequency and is defined as \cite{SoumyaITPWRS2022, pico2022blackstart}:
\begin{align*}
\omega_c\hspace{-0.09cm} = \hspace{-0.09cm}\dfrac{\sum_{i=1}^G \omega_{g}(i)H_{g}(i)\hspace{-0.09cm} + \hspace{-0.09cm}\sum_{j=1}^S \omega_{s}(j)H_{s}(j)\hspace{-0.09cm} +\hspace{-0.09cm} \sum_{l=1}^W \omega_{w}(l)H_{w}(l)}{\sum_{i=1}^G H_{g}(i) + \sum_{j=1}^S H_{s}(j) + \sum_{l=1}^W H_{w}(l)}
\end{align*}
with $H_{g}, H_{s}$, and $H_{w}$ representing the inertial constants of synchronous machines, solar, and wind power plant, respectively.
\begin{figure}
	\hspace{-0.1cm}\subfloat{\includegraphics[keepaspectratio=true,scale=0.53]{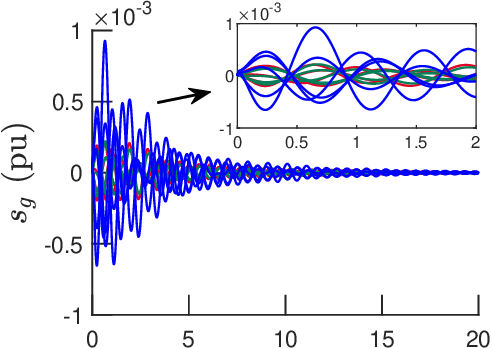}}{}\hspace{-0.05cm}
	\subfloat{\includegraphics[keepaspectratio=true,scale=0.53]{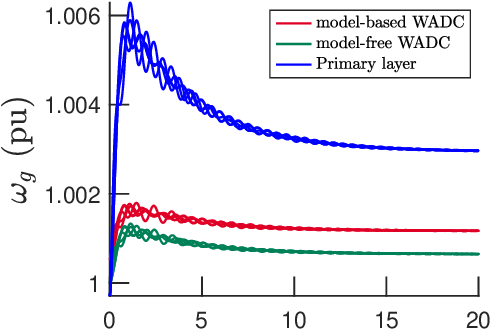}}{}{}	
	\subfloat{\includegraphics[keepaspectratio=true,scale=0.53]{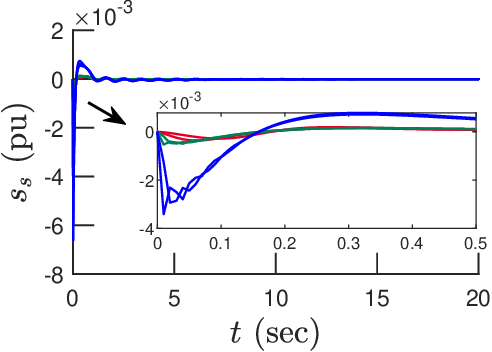}}{}{}\hspace{-0.05cm}
	\subfloat{\includegraphics[keepaspectratio=true,scale=0.53]{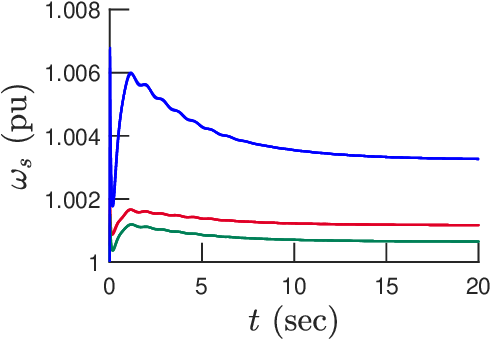}}{}{}\vspace{-0.69cm} \caption{Comparative analysis under $\Delta_L = -0.003$ and $\Delta_{I_s} = 0.1$ for 39-bus test system: all synchronous generators slip and rotor frequency (above),  slip and relative speed of all  solar power plants (below).}\label{fig:pd dec case2 39}
	\vspace{-0.5cm}
\end{figure}
That being said, we can also notice that for both test systems, there is a noticeable improvement in frequency nadir (the overall frequency dip or increase during the transient period). For 39-bus system, from Fig. \ref{fig:pd dec case2 39}  with only primary control layer after disturbance the synchronous machines frequency rises upto almost $1.006$pu while with proposed  WADCs it is limited only to around $1.0019$pu thus improving the overall system frequency nadir. Similar results are achieved for the inverter-based resources, we can see that there is a significant improvement in the relative frequency nadir for both solar and wind power plants.
These results are also corroborated  from Tab. \ref{tab:comp_speed_deviation} where the overall weighted-average frequency deviation of the power system is presented, we can see that for both test systems there are significant improvements under all types of disturbances.

\vspace{-0.3cm}
\subsection{Discussion on model-based vs model-free feedback control for the future power grid}
\vspace{-0.02cm}
Notice that in both proposed approaches the main difference lies in the fact that in the case of model-free design, the controller only requires realtime information of state vector $\m x$, which these days can be obtained accurately as there exist highly robust state estimation algorithms. These estimation algorithms based on measurements received from a few PMUs can accurately estimate all the states including the states of solar/wind, composite loads, and synchronous machines \cite{nadeem2022robust}. While for the model-based approach, we not only need the information of state vector $\m x$ but also accurate information about the system (matrices $\m A$, $\m B_u$, and so on) in realtime, which can be highly problematic for a larger power system and with high penetration of uncertain renewables. \textcolor{black}{Hence whenever any large parametric uncertainty happens the model-based controller must be recalibrated to maintain reliable performance.} On the other hand, although the model-free approach seems reasonable for the future power grid as knowledge of the system is not required, however, the long training time, tuning of hyper-parameters, and no solid theoretical stability certificate/guarantee of the learned policy (since it merely solves an optimization problem and maximizes the long term reward) can be problematic. In the case of  model-based approach, if there is a solution to Proposition \ref{theorm:H_inf}, then, there is a solid theoretical guarantee (based on the Lyapunov stability notion) that obtained control policy will asymptotically stabilize the perturbed system (or in other, words push the original power system model back to its original steady state equilibrium).
	
\textcolor{black}{Furthermore, notice that, according to the control theory literature, the model-based feedback control approaches can easily be extended to other more robust feedback control designs without adding significant complexities to the controller architecture, such as $\mathcal{H}_\infty$ or $\mathcal{L}_\infty$-type feedback controller design which has shown to have superior performance \cite{YuFend, taha2019robust}. While in the case of RL-based model-free designs, one can only maximize the long-term reward (since no information about the system is available) then it might be even much more harder to solve robust feedback controller ($\mathcal{H}_\infty$, $\mathcal{L}_\infty$ and so on) design as compared to their model-based counterparts.}

\textcolor{black}{Moreover, in the case of model-based designs, if the LMI-based formulations (which are used in this paper) become numerically unstable (or ill-conditioned, which they sometimes become depending on the input system matrices)---see \cite{lofberg2009pre}, then, they can also be formulated as continuous-time algebraic Riccati equations (CAREs) formats which can be solved highly efficiently and there exist well-built off-the-shelf software's/tools available to solve these types of feedback control problems easily \cite{LiTAC1993, CloutierACC2002, AranyaICSM2019}.  In short, these are interesting future research avenues about how to efficiently and tractably solve (similar to the model-based approach) the optimal feedback (and robust optimal feedback) control problems via RL-based model-free approach for larger detailed renewables heavy power system models.}

\textcolor{black}{Finally, we also want to point out here that the RL-based approach requires about 6–9 hours of training for the IEEE 9-bus and 39-bus systems because it implements a dense, full-state controller that coordinates all generation sources (synchronous generators, solar farms, and wind plants) following disturbances. A practical way to shorten training time is to promote sparsity or adopt a distributed/decentralized method. By comparison, the model-based approach is faster to set up—taking roughly 2 seconds for the 9-bus system and 54 seconds for the 39-bus system—but again requires an accurate system model and the ability to handle nonlinearities. Both methods remain viable for real-world applications, as modern measurement technologies (e.g., PMUs) and high-performance estimation algorithms \cite{SebastianITPWRS2020, SebastianITCNS2021, Liu2021TPWRS} now make full-state feedback achievable with only a few well-placed sensors \cite{NadeemCDC2022, NadeemITPWRS2022}.}

\begin{table}
\vspace{-0.2cm}
	\caption{Comparison of Computational time between the proposed two approaches.}
	\label{tab:Table 1}
	\setlength{\tabcolsep}{1.5pt}
	\small
	\centering
	\begin{tabular}{|c|cc}
		\hline
		\begin{tabular}[c]{@{}c@{}} Network \end{tabular} &
		\multicolumn{2}{c|}{\begin{tabular}[c]{@{}c@{}} Time required to compute the \\ optimal control policy \end{tabular}} \\ \hline \hline
		& \multicolumn{1}{c|}{model-based WADC} & \multicolumn{1}{c|}{model-free WADC}\\ \hline
		 $9$-bus test system  & \multicolumn{1}{c|}{2.2 seconds}                      & \multicolumn{1}{c|}{6.28 hrs}   \\ \hline
		 $39$-bus test system & \multicolumn{1}{c|}{53.2 seconds}                      & \multicolumn{1}{c|}{8.32 hrs}     \\ \hline
	\end{tabular}
\end{table}

\setlength{\textfloatsep}{10pt}
\begin{table}
	\vspace{-0.1cm}
		\setlength{\tabcolsep}{1.5pt}
	\small
	\centering 
	\caption{Comparative analysis of the deviation of overall weighted-average system frequency for the whole power system under various disturbances in load and renewable generations.}\label{tab:comp_speed_deviation}
	\renewcommand{\arraystretch}{1.65}
	\begin{threeparttable}
		\begin{tabular}{|c|c|c|c|c|}
			\hline
			\multirow{2}{*}{Network} & \multirow{2}{*}{Disturbance} & \multicolumn{3}{c|}{$\left({\omega_{0}-  \omega_c} \right)$}                                                 \\ \cline{3-5} 
			&    &   Primary            &   model-based    &  model-free  \\ \hline \hline
			\multirow{3}{*}{$9$-bus} &      $\Delta_L = 0.7$           &  ${0.015}$     &  ${0.004}$     & ${0.0036}$      \\ \cline{2-5} 
			&      $\Delta_L = -0.5$     &  $-0.013$     &  $-0.0037$     &   $-0.0038$      \\ \cline{2-5} 
			&      $\Delta_L\hspace{-0.09cm} =\hspace{-0.09cm} 0.6$, $\Delta_{I_s}\hspace{-0.09cm} =\hspace{-0.09cm} 0.1$                  &   ${0.019}$  &    $0.0027$    &   $0.0029$     \\ \hline
			\multirow{3}{*}{$39$-bus} &      $\Delta_L = 0.003$            &  ${2.6\hspace{-0.09cm}\times\hspace{-0.09cm} 10^{-3}}$     &   ${1.3\hspace{-0.09cm}\times\hspace{-0.09cm} 10^{-5}}$    &  ${1.9\hspace{-0.09cm}\times\hspace{-0.09cm} 10^{-5}}$     \\ \cline{2-5} 
			&       $\Delta_L\hspace{-0.09cm} =\hspace{-0.09cm} -0.001$        &  ${-1\hspace{-0.09cm}\times\hspace{-0.09cm} 10^{-3}}$     &    ${-2\hspace{-0.09cm}\times\hspace{-0.09cm} 10^{-5}}$   &   ${-2\hspace{-0.09cm}\times\hspace{-0.09cm} 10^{-5}}$   
			\\ \cline{2-5} 
			&     $\Delta_L\hspace{-0.09cm} =\hspace{-0.09cm} 0.04$, $\Delta_{I_s}\hspace{-0.09cm} =\hspace{-0.09cm} 0.1$         &  ${7.9\hspace{-0.09cm}\times\hspace{-0.09cm} 10^{-3}}$ &   ${1.9\hspace{-0.09cm}\times\hspace{-0.09cm} 10^{-4}}$    &  ${2.1\hspace{-0.09cm}\times\hspace{-0.09cm} 10^{-4}}$       \\ \hline
		\end{tabular}
	\end{threeparttable}
	\vspace{-0.1cm}
\end{table}
\setlength{\floatsep}{10pt}
\vspace{-0.4cm}
\section{Concluding Remarks}\label{sec: conclusion}
\vspace{-0.4cm}
\textcolor{black}{In this paper, we addressed the optimal feedback control problem for renewable-heavy power systems by modeling detailed solar, wind, and composite load dynamics. We explored two distinct approaches: a \emph{model-free} design using DDPG-based reinforcement learning and a \emph{model-based} method grounded in Lyapunov theory. Our simulations on the 9-bus and 39-bus IEEE systems show that adding an extra wide-area feedback control loop can significantly improve transient stability after large disturbances. We also performed a thorough comparison of these methods, highlighting their advantages and disadvantages. Importantly, we do not claim to offer a universal solution or a strict guideline on whether to choose model-free or model-based control. Instead, we aimed to illustrate how each strategy behaves, analyze their pros and cons, and suggest future research directions. We hope these insights will be valuable to both researchers and industry practitioners working in designing feedback control strategies for power grids.}

The limitations of the presented work are as follows:  Firstly, for the both proposed WADCs, the designed control policy is dense (requires all the power plant to take part in the control action) and not sparse. \textcolor{black}{Secondly, both approaches does not take into account delays and/ or cyber-attacks in the communication network.} Lastly, for the proposed model-free WADCs, the long training time and no solid theoretical stability guarantee (similar to its model-based counterpart which guarantees stability using Lyapunov criterion) are the drawbacks. Future work will be about addressing the aforementioned limitations and designing stability-aware tractable model-free WADC.
\vspace{-0.5cm}
\bibliographystyle{unsrt}
\bibliography{mybibfile}
\vspace{-0.5cm}

\appendix

\section{Details of considered power system model}\label{appndix:ninth Gen_dynamics}
\noindent In  the NDAE \eqref{equ:PSModel}, $\m x_a(t) \in \mbb{R}^{n_a}$ is given as:
\begin{align}\label{eq:x_a}
	\m x_a(t)  = \bmat{\m V_{\mr{Im}}^\top&\m V_{\mr{Re}}^\top&\m I_{\mr{Im}}^\top& \m I_{\mr{Re}}^\top&}^\top
\end{align}
where $ \m V_{\mr{Im}}\hspace{-0.0cm}=\hspace{-0.05cm} \hspace{-0.0cm}\{V_{\mr{Im}_i}\}_{i\in \mc{N}} $ represent the imaginary part of voltages and $\m V_{\mr{Re}}\hspace{-0.1cm}= \hspace{-0.051cm}\{V_{\mr{Re}_i}\}_{i\in \mc{N}}$ are the real parts of voltage phasors. Similarly, $ \m I_{\mr{Im}}\hspace{-0.1cm}= \hspace{-0.1cm}\{I_{\mr{Im}_i}\}_{i\in \mc{N}}, \m I_{\mr{Re}}\hspace{-0.1cm}= \hspace{-0.1cm}\{I_{\mr{Re}_i}\}_{i\in \mc{N}}$ are the imaginary and real parts of current phasors. The input vector $\m u(t)\in \mbb{R}^{n_u}$ lumps the control inputs for all the power plants and is given as:
\vspace{-0.1cm}
\begin{align}\label{eq:u}
	\m u(t) = \bmat{\m u_G^\top & \m u_S^\top& \m u_W^\top}^\top
\end{align}
where $\m u_G(t) = \bmat{\m P_v^{*\top} & \m V_g^{*\top}}^\top \in \mbb{R}^{2G}$ represents the control inputs of synchronous generators with $\m V_g^*$ denoting voltage set points of AVRs of the generators and $\m P_v^*$ representing turbine valve position set-points. Similarly, $ \m u_S(t) = \bmat{\m P_s^{*\top} &\m V_s^{*\top}}^\top \in \mbb{R}^{2S}$ and  $ \m u_W(t) = \bmat{\m P_w^{*\top}& \m V_w^{*\top}}^\top \in \mbb{R}^{2W}$ are the control inputs for solar and wind power plants 
with $\m P_s^*$, $\m P_w^*$ and  $\m V_s^*$,  $\m V_w^*$ denoting active power and voltage set-points for solar and wind-based power plants, respectively.

Also, in the system dynamics \eqref{equ:PSModel} the disturbance vector $\m w(t)\in \mbb{R}^{n_w}$ is modeled as  $\m w(t) = \bmat{\m I_{s}^\top& \m P_{L}^\top}^\top$ in which  $\m I_s$ is the solar irradiance $(W/m^2)$ and $\m P_L$ represent the system active power load demand.  

Moreover, the vector $\m x_d(t)\in \mbb{R}^{n_d}$ in \eqref{equ:PSModel} lumps the dynamic states of traditional power plant, solar, wind-based plant, and loads which is expressed as:
\begin{align}\label{eq:x_d}
	\m x_d(t) = \bmat{\m x_G^\top&\m x_S^\top&\m x_W^\top&\m x_M^\top}^\top
\end{align} 
where $\m x_G(t)$ represents the dynamic states of the conventional power plant, $\m x_S(t)$ lumps dynamic states of solar power plant, $\m x_W(t)$ contains wind power plant states, and $\m x_M(t)$  denotes the dynamic states of the motor-based loads. 

That being said, we model the conventional/traditional power plant via a comprehensive $9^{th}$-order dynamical model representing generators swing equations, excitation system, governor, and
turbine dynamical models. Then, $\m x_G(t) \in\hspace{-0.01cm} \mbb{R}^{9G}$ is  represented as follows \cite{sauer2017power,SoumyaITPWRS2022}:
\begin{align*}
	\m x_G(t)\hspace{-0.0cm} = \hspace{-0.0cm}\bmat{ \m \delta_{\mr{g}}^\top\,\,\m\omega_{\mr{g}}^\top \,\, \m E_{\mr{fd}}^\top\,\,\,  \m E_{\mr q}^\top\,\,  \m E_{\mr d}^\top\,\,\m T_\mr{M}^\top\,\,\, \m P_{v}^\top\,\, \m r_{f}^\top\,\,\m v_{a}^\top}^\top \hspace{-0.2cm} \label{eq:stateSyncGen}
\end{align*}
where  $\m \delta_{\mr g}$ denotes the rotor angle of the generator, $\m \omega_{\mr g}$ is the generator rotor speed, $\m E_{\mr{fd}}$ is the field voltage of the generator,  $\m E_{\mr q}$, $\m E_{\mr d}$ are the transient voltages of the generator along q-axis and d-axis, respectively, $\m T_{\mr M}$ represents torque of the prime mover of the turbine, $\m P_{v}$ denotes turbine valve position,  while $\m r_{f}$ and $\m v_{a}$  are the stabilizer output and amplifier voltages, respectively. For further detailed explanations about the generator dynamics used in this study readers are referred to \cite{sauer2017power}.

In \eqref{eq:x_d}, we model the solar plant dynamics via  $12^{th}$-order dynamical system representing solar plants working in grid-forming mode as detailed in \cite{SoumyaITPWRS2022, WasynczukITPE1996, nadeem2023RL}. The complete dynamical equations modeling solar farms represents, DC side dynamics modeling PV arrays, DC link, current/voltage regulators dynamic equation, and  AC side dynamics representing LCL filter equations, and AC/DC converter dynamics. That being said, the state vector for the solar plants  $\m x_S(t) \in \hspace{-0.05cm}\mbb{R}^{12S} $ can be written as follows:
\begin{align*}
\m x_S(t)\hspace{-0.0cm} =\hspace{-0.0cm} \bmat{\m E_{\mr{dc}}^\top\,\m P_{s}^\top\,\m Q_{s}^\top\, \m \delta_{\mathrm{s}}^\top\,\m i_{\mr{df}}^\top\,\,\m i_{\mr{qf}}^\top\,\m v_{\mr{dc}}^\top \,\,\m v_{\mr{qc}}^\top \,\,\m z_{\mr{df}}^\top \,\,\m z_{\mr{qf}}^\top \,\, \m z_{\mr{do}}^\top\, \m z_{\mr{qo}}^\top}^\top
\end{align*}
\noindent where $\m{E}_{\mathrm{dc}}$ denotes the energy stored in the DC side capacitor,  $\m{{P}}_{s}$ is the real power while $\m{{Q}}_{s}$ represents the reactive power injected by solar plants to the power grid, $\m{\delta}_{{s}}$ is the  solar plant relative angle, $\m{i}_{\mr{df}}, \m{i}_{\mr{qf}}$, $\m{v}_{\mr{df}}, \m{v}_{\mr{qf}}$ are the dq-axis current and voltages from the solar power plants at their terminal bus, respectively, while  $\m{z}_{\mr{df}}$,$\m{z}_{\mr{qf}}$,$\m{z}_{\mr{do}}$,$\m{z}_{\mr{do}}$ are the dynamic states of voltage and current regulators along dq-axis, respectively.  Interested readers are referred to \cite{SoumyaITPWRS2022, WasynczukITPE1996} for further in-depth explanations of the solar power plant model considered in this study. 

Similarly, the dynamics of wind power plants have been modeled via $13^{th}$-order dynamical system as detailed in \cite{pico2022blackstart,pico2014voltage}. The overall model describes double-fed induction generator (DFIG)-based wind turbine acting in a grid forming-mode and thus the state vector for the wind plants $\m x_W(t) \in \hspace{-0.05cm}\mbb{R}^{13W}$ can be expressed as follows:
\begin{align}
	\hspace{-0.0cm}\m x_W(t)\hspace{-0.0cm} =\hspace{-0.0cm} \bmat{\m \delta_{\mathrm{w}}^\top\,\,\m E_{\mr{dcw}}^\top\,\,\m i_{\mr{wf}}^\top\,\, \m v_{\mr{wc}}^\top \,\,\m P_{w}^\top \,\,\m Q_{w}^\top\,\, \m z_{\mr{dc}}^\top\,\,\m z_{\mr{wf}}^\top \,\, \m z_{\mr{wo}}^\top}^\top
\end{align}
where $\m{\delta}_{{w}}$ represents the inverter  relative angle of the wind plant, $\m{E}_{\mathrm{dcw}}$ is the energy stored in the DC link capacitor of the wind plant, $\m{i}_{\mr{wf}}=[\m i_\mr{dw_f}^\top\,\,\m i_\mr{qw_f}^\top]^\top$ denotes the inverter output current at the terminal bus along dq-axis, similarly $\m{v}_{\mr{wc}} \hspace{-0.15cm}=\hspace{-0.02cm} [\m v_\mr{dw_c}^\top\,\,\m v_\mr{qw_c}^\top]^\top$ represents the dq-axis AC capacitor voltages, $\m{{P}}_{w}$, $\m{Q}_{w}$ respectively denote the  real and reactive power output from the wind power plant to the power grid, while $\m{z}_{\mr{dc}}$, $\m{z}_{\mr{wf}}\hspace{-0.1cm}=\hspace{-0.04cm}[\m z_\mr{dw_f}^\top\,\,\m z_\mr{qw_f}^\top]^\top$, $\m{z}_{\mr{wo}}\hspace{-0.1cm}=\hspace{-0.04cm}[\m z_\mr{dw_0}^\top\,\,\m z_\mr{qw_0}^\top]^\top$, are the states of the current and voltage regulators of the wind power plant along dq-axis, respectively. 
Further details explanations about the wind power plants used in this work can be found in \cite{pico2022blackstart,pico2014voltage}.


Furthermore, in \eqref{eq:x_d} $\m x_M(t) \in \mbb{R}^{L_k}$ denotes the speed $\omega_{\mr m_i}(t)$ of the motor-based load and is given as \cite{krause2013}:
\vspace{-0.0cm}
\begin{equation}
	\dot{\omega}_{\mr{m}_i}(t) = \frac{1}{2H_{\mr M_i}}(T_{{e_i}} - T_{\mr M_i}) \label{eq:omegam}
\end{equation}
where $i\in \mc{L}$, $T_{{e_i}}$ represent electromagnetic torque of motor, $T_{{M_i}}$ is  mechanical torque while $H_{\mr M_i}$ is the inertia constant of the motor-based load \cite{krause2013}. 
The constant impedance/power-based loads satisfy the following equations \cite{NERC}: 
\vspace{-0.0cm}
\begin{subequations}
	\begin{align}
		\begin{split}\label{eq:load_dyn_z}
			I_{z_i}Z_i+ V_{z_i}  &= 0
		\end{split}\\
		(P_{p_i}+Q_{p_i})+ \mr{conj}(I_{p_i})V_{p_i} &= 0
	\end{align}
\end{subequations}
where $i\in \mc{L}$,  $Z_i$ represent the impedance of the load, $P_{p_i}$, $Q_{p_i}$ are the real and reactive power of constant power loads, the term $\mr{conj}$ denotes the complex conjugate operator while $I_{z_i}$, $V_{z_i}$,   $I_{p_i}$, $V_{p_i}$ are the current and voltage phasors of buses connected to constant impedance and constant power loads, respectively. This completes the modeling of differential equations  \eqref{equ:PSModel-a} of the power system NDAE model \eqref{equ:PSModel}.

The algebraic constraints \eqref{equ:PSModel-b}  model the current balance equation between all the power plants and loads and thus can be written as \cite{sauer2017power}: 
\begin{gather}
	\underbrace{\begin{bmatrix}
			\m{{I}}_{G}(t) \\ \m{{I}}_{S}(t) \\ \m{{I}}_{L}(t)\\ \m{{I}}_{W}(t)
	\end{bmatrix}}_{\m I(t)}
	\hspace{-0.0cm}-\hspace{-0.0cm}
	\underbrace{\begin{bmatrix}
			\m{Y}_{GG} \,\,\, \m{Y}_{GS} \,\,\,\, \m{Y}_{GL} \,\,\,\, \m{Y}_{GW}\\
			\m{Y}_{SG} \,\,\,\, \m{Y}_{SS} \,\,\,\, \m{Y}_{SL} \,\,\,\, \m{Y}_{SW}\\
			\m{Y}_{LG} \,\,\,\, \m{Y}_{LS}\,\,\,\, \m{Y}_{LL} \,\,\,\, \m{Y}_{LW}\\
			\m{Y}_{WG} \,\, \m{Y}_{WS} \,\, \m{Y}_{WL} \,\, \m{Y}_{WW}
	\end{bmatrix}}_{\m Y}\hspace{-0.01cm}
	\underbrace{\begin{bmatrix}
			\m{{V}}_{G}(t) \\ \m{{V}}_{S}(t) \\ \m{{V}}_{L}(t) \\ \m{V}_{W}(t)
	\end{bmatrix}}_{\m V(t)}\hspace{-0.04cm} =\hspace{-0.04cm}  \m{0} \label{eq:transalgebraic}
\end{gather}
\noindent where $\m V(t)$, $\m I(t)$  are the net voltages and currents and $\m{Y}$ represents the system admittance matrix. In \eqref{eq:transalgebraic} the terms $\m{{I}}_{G}\hspace{-0.13cm}= \hspace{-0.13cm}\{I_{Re_i}\}_{i\in \mc{G}}\hspace{-0.05cm}+\hspace{-0.05cm}j \{I_{Im_i}\}_{i\in \mc{G}}\hspace{-0.05cm}$, $\m{{V}}_{G} \hspace{-0.06cm}= \hspace{-0.06cm}\{V_{Re_i}\}_{i\in \mc{G}}\hspace{-0.05cm}+\hspace{-0.05cm} j\{V_{Im_i}\}_{i\in \mc{G}}\hspace{-0.05cm}$,  denote real and imaginary parts of synchronous generators terminal buses current and voltage phasors, respectively. Similarly, $\m{{I}}_{S}$, $\m{{I}}_{W}$, $\m{{I}}_{L}$ and $\m{{V}}_{S}$, $\m{{V}}_{W}$, $\m{{V}}_{L}$ denote current and voltage phasors of all solar, wind-based power plants, and loads, respectively. 
\begin{figure}[htp]
	\centering
	{\includegraphics[keepaspectratio,scale=0.1]{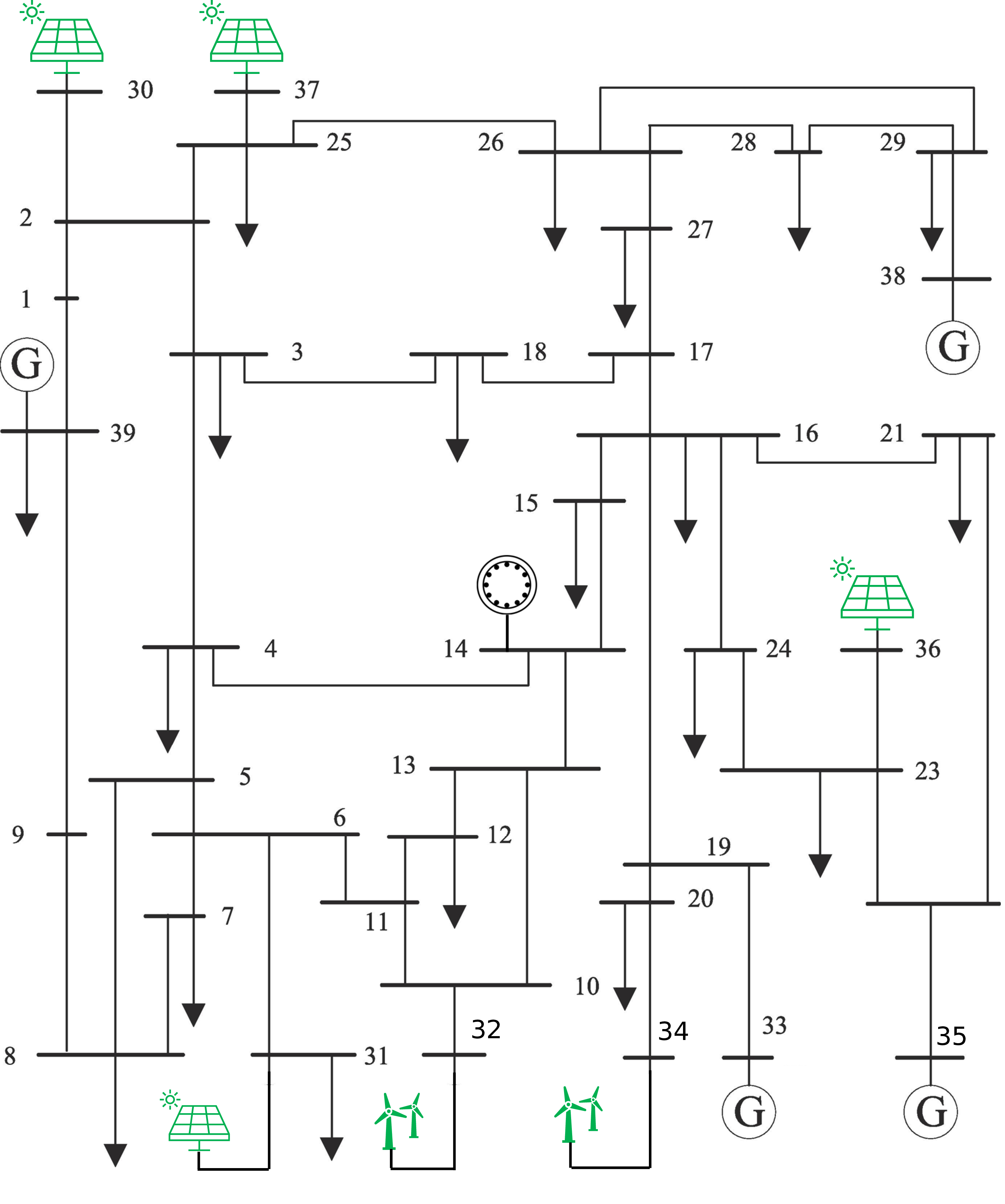}}\caption{Diagram of the 39-bus test system with PV plants at Buses $30$, $31$, $36$, $37$, wind-based plants at Buses $32$, $34$, conventional power plants at Buses $33$, $35$, $38$, $39$, and a motor-based load at Bus $14$.}\label{fig:Diagram39}
\end{figure}
\begin{figure}[htp]
	\centering
{\includegraphics[keepaspectratio,scale=0.09]{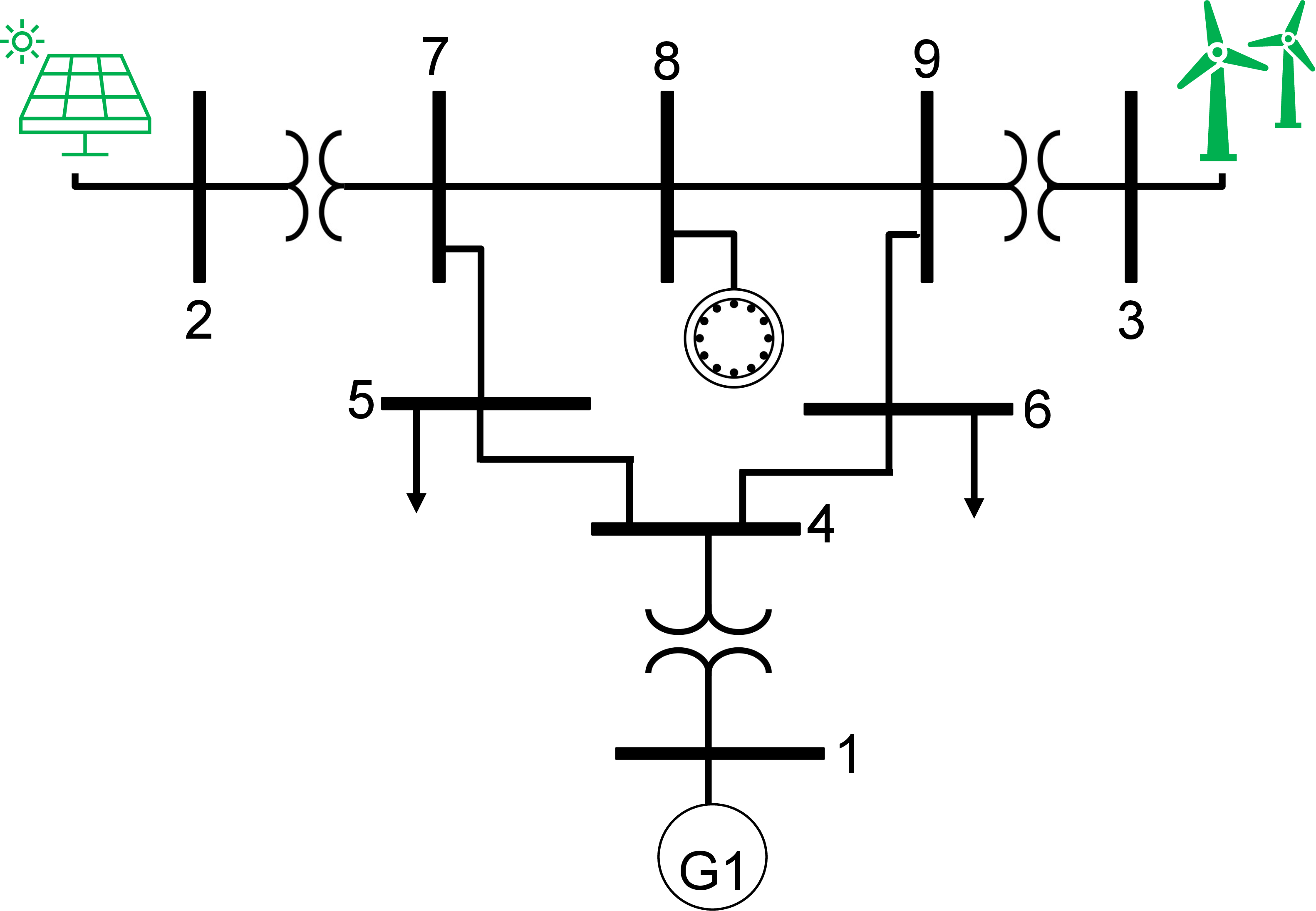}}\vspace{-0.2cm}\caption{Diagram of the 9-bus system with a PV-based plant at Bus $2$, a wind plant at Bus $3$, steam-based plant at Bus $1$, and a motor-based load at Bus $8$.}\label{fig:Diagram}
\end{figure}
\end{document}

%% file: preamble_eslvier.tex
\usepackage{amsthm} 
\usepackage[intlimits,fleqn]{amsmath}    
\usepackage[T1]{fontenc}
\usepackage[utf8]{inputenc}
\biboptions{sort&compress}
\usepackage{multirow}

\usepackage{bm}
\usepackage{epstopdf}
\usepackage{amssymb}
\usepackage{url}
\usepackage{enumitem} 
\usepackage{multirow}
\usepackage{hhline}
\usepackage{booktabs}
\usepackage{mathtools}
\usepackage{makecell}
\usepackage[linesnumbered,boxed,commentsnumbered,ruled,vlined,longend]{algorithm2e}
\usepackage{comment}

\DeclareMathOperator*{\subjectto}{subject\ to}

\makeatother
\DeclareMathAlphabet\mathbfcal{OMS}{cmsy}{b}{n}

\newtheorem{asmp}{Assumption}

\newtheorem{myprs}{Proposition}



\usepackage{stackengine}

\newcommand{\mat}[1]{\boldsymbol{#1}}

\newcommand{\bmat}[1]{\begin{bmatrix} #1 \end{bmatrix}}

\providecommand{\mA}{\ensuremath{\mat{A}}}
\providecommand{\mB}{\ensuremath{\mat{B}}}
\providecommand{\mC}{\ensuremath{\mat{C}}}
\providecommand{\mD}{\ensuremath{\mat{D}}}

\providecommand{\mF}{\ensuremath{\mat{F}}}

\providecommand{\mI}{\ensuremath{\mat{I}}}
\providecommand{\mJ}{\ensuremath{\mat{J}}}
\providecommand{\mK}{\ensuremath{\mat{K}}}

\providecommand{\mO}{\ensuremath{\mat{O}}}
\providecommand{\mP}{\ensuremath{\mat{P}}}
\providecommand{\mQ}{\ensuremath{\mat{Q}}}
\providecommand{\mR}{\ensuremath{\mat{R}}}

\providecommand{\mW}{\ensuremath{\mat{W}}}

\providecommand{\mZ}{\ensuremath{\mat{Z}}}

\newcommand{\m}{\boldsymbol}
\allowdisplaybreaks[4]
\usepackage[colorlinks = true,
linkcolor = blue,
urlcolor  = blue,
citecolor = blue,
anchorcolor = blue]{hyperref}

\newcommand{\mc}[1]{\mathcal{#1}}
\newcommand{\mbb}[1]{\mathbb{#1}}
\newcommand{\mr}[1]{\mathrm{#1}}
\usepackage[framemethod=TikZ]{mdframed}
\mdfdefinestyle{MyFrame}{%
	linecolor=black,
	outerlinewidth=1.25pt,
	roundcorner=1.25pt,
	innerrightmargin=5pt,
	innerleftmargin=5pt,}

\usepackage[noabbrev]{cleveref}

\usepackage{mathtools}

\DeclarePairedDelimiter\abs{\lvert}{\rvert}%
\DeclarePairedDelimiter\norm{\lVert}{\rVert}%

\makeatletter
\let\oldabs\abs
\def\abs{\@ifstar{\oldabs}{\oldabs*}}
\let\oldnorm\norm
\def\norm{\@ifstar{\oldnorm}{\oldnorm*}}
\makeatother

\usepackage[english]{babel}
\usepackage[utf8]{inputenc}
\usepackage[super]{nth}

\usepackage{graphicx}
\usepackage{float}
\usepackage[caption = false]{subfig}

\usepackage{array}
\usepackage{threeparttable}

\usepackage[english]{babel}
\usepackage[utf8]{inputenc}
\usepackage[super]{nth}

\SetKwRepeat{Do}{do}{while}%

